\documentclass[11pt]{article}
\usepackage{graphicx,amsthm,fancyhdr,mathrsfs}
\usepackage{amsfonts}
\usepackage{amsmath}
\usepackage{amssymb}
\usepackage{url}
   \usepackage[colorlinks=true,citecolor=blue]{hyperref}
\usepackage{indentfirst}
\usepackage{enumerate}
\usepackage{color}
\usepackage{comment}
\usepackage{dsfont}
\usepackage{natbib}
\usepackage[misc]{ifsym}
\usepackage[toc,page]{appendix}
\usepackage[below]{placeins} 
\usepackage{float}
\usepackage{caption} 
\usepackage{subcaption} 
\usepackage{graphicx,psfrag,epsf}
\usepackage{enumerate}
\usepackage{natbib}
\usepackage{setspace}
\usepackage{comment}

\def\d{\mathrm{d}}
\def\laweq{\buildrel \d \over =} 
\def\pto{\buildrel \mathrm p \over \rightarrow}
\def\dto{\buildrel \d \over \rightarrow}

\newcommand{\var}{\mathrm{Var}}
\newcommand{\cov}{\mathrm{Cov}}

\newcommand{\SD}{\mathrm{SD}}
\newcommand{\VaR}{\mathrm{VaR}}

\newcommand{\ES}{\mathrm{ES}}
\newcommand{\ex}{\mathrm{ex}}
\newcommand{\E}{\mathbb{E}}
\newcommand{\R}{\mathbb{R}}

\newcommand{\Gini}{\mathrm{Gini}}

\newcommand{\N}{\mathbb{N}}
\newcommand{\p}{\mathbb{P}}

\newcommand{\id}{\mathds{1}}

\newcommand{\X}{\mathcal X}

\renewcommand{\(}{\left(}
\renewcommand{\)}{\right)}
\renewcommand{\[}{\left[}
\renewcommand{\]}{\right]}

\newcommand{\esssup}{\mathrm{ess\mbox{-}sup}}
\newcommand{\essinf}{\mathrm{ess\mbox{-}inf}}
\renewcommand{\ge}{\geqslant}
\renewcommand{\le}{\leqslant}
\renewcommand{\geq}{\geqslant}
\renewcommand{\leq}{\leqslant}
\renewcommand{\epsilon}{\varepsilon}

\usepackage{array}
\newcommand{\PreserveBackslash}[1]{\let\temp=\\#1\let\\=\temp}
\newcolumntype{C}[1]{>{\PreserveBackslash\centering}p{#1}}
\newcolumntype{R}[1]{>{\PreserveBackslash\raggedleft}p{#1}}
\newcolumntype{L}[1]{>{\PreserveBackslash\raggedright}p{#1}}

\usepackage{booktabs,array}

\newcount\rowc

\theoremstyle{plain}
\newtheorem{theorem}{Theorem}

\newtheorem{proposition}{Proposition}
\theoremstyle{definition}
\newtheorem{definition}{Definition}
\newtheorem{example}{Example}

\newtheorem{remark}{Remark}

\usepackage{tikz}

\begin{document}

\renewcommand{\baselinestretch}{1.05}
\textheight 22.23cm


\title{\bf Parametric measures of variability induced by risk measures}

\author{Fabio Bellini\thanks{Department of Statistics and Quantitative Methods, University of Milano-Bicocca.  \newline   \indent \indent   Piazza dell'Ateneo Nuovo, 1, 20126, Milan, Italy.    \hfill \Letter~\texttt{fabio.bellini@unimib.it}} \and Tolulope Fadina\thanks{Corresponding author. Department of Mathematical Sciences, University of Essex. \newline  \indent  \indent   Wivenhoe Park, Colchester, CO4 3SQ, United Kingdom. \hfill \Letter~\texttt{t.fadina@essex.ac.uk} } 
 \and Ruodu Wang\thanks{Department of Statistics and Actuarial Science, University of Waterloo.    \newline   \indent   \indent  200 University Avenue W., Waterloo ON, N2L 3G1, Canada.     \hfill \Letter~\texttt{wang@uwaterloo.ca}  }  \and Yunran Wei\thanks{School of Mathematics and Statistics, Carleton University. \newline  \indent \indent 1125 Colonel By Dr., Ottawa ON, K1S 5B6, Canada.  \hfill \Letter~\texttt{yunran.wei@carleton.ca}   } }

\date{\today} 
\maketitle


%

\bigskip
\begin{abstract}
We present a general framework for a comparative theory of variability measures, with a particular focus on the recently introduced one-parameter families of inter-Expected Shortfall differences and inter-expectile differences, that are explored in detail and compared with the widely known and applied inter-quantile differences. 

From the mathematical point of view, our main result is a characterization of symmetric and comonotonic variability measures as mixtures of inter-Expected Shortfall differences, under a few additional technical conditions. Further, we study the stochastic orders induced by the pointwise comparison of inter-Expected Shortfall and inter-expectile differences, and discuss their relationship with the dilation order. From the statistical point of view, we establish asymptotic consistency and normality of the natural estimators and provide a rule of the thumb for cross-comparisons. 

Finally, we study the empirical behaviour of the considered classes of variability measures on the S\&P $500$ Index under various economic regimes, and explore the comparability of different time series according to the introduced stochastic orders.  
\end{abstract}

\noindent%
\vfill

\section{Introduction}
\label{sec:intro}

Several measures of distributional variability are widely used in statistics, probability, economics, finance, physical sciences, and other disciplines. 
In this paper, we study a general theory of variability measures with an emphasis on three symmetric one-parameter families generated by popular parametric risk measures: Value-at-Risk (VaR),  Expected Shortfall (ES), and expectiles. The corresponding induced variability measures 
are the \emph{inter-quantile difference}, the \emph{inter-ES difference}, and the \emph{inter-expectile difference}. While the first one is a classical measure of statistical dispersion widely used e.g.~in box plots, 
the other two  are, to the best of our knowledge, relatively new: 
the inter-ES difference appears in Example 4 of \cite{WWW20b} as a signed Choquet integral, and the inter-expectile difference has been studied in \cite{BMR20} via a connection to option prices. The present paper is a first unifying study, focused on their comparative qualitative and quantitative properties.  

The mathematical theory of risk measures is extensive, and a standard reference is \cite{FS16}. 
As it is well-known, VaR is simply a quantile and ES is a coherent risk measure in the sense of \cite{ADEH99}. Both VaR and ES are implemented in current banking and insurance regulation frameworks; we refer to \cite{MFE15} for a comprehensive background and  \cite{WZ20} for a more recent account. Expectiles, originally introduced in the statistical literature by \cite{NP87}, have received an increasing attention in risk management, as it has been shown that they are the only elicitable coherent risk measures (\cite{Z16}).
We refer e.g.~to \cite{BKMR14}  and \cite{BD15}  for more on the theory and financial applications of expectiles. For a comparison of the above risk measures in the context of regulatory capital calculation, see \cite{EPRWB14} and \cite{EKT15}.

The theory of variability measures has been studied from different angles; see \cite{D98} for a review in the context of the measurement of statistical dispersion. A  mathematical formulation closer to our setting is the notion of deviation measure introduced in \cite{RUZ06}, and further developed by \cite{GMZ09, GMZ10}. A similar notion of variability measure was proposed by \cite{FWZ17} with an emphasis on the Gini deviation. We will explain in Section \ref{sec2} the differences between our general definition and the ones given in the literature; in particular, the inter-quantile difference does not satisfy the definition of deviation measure of \cite{RUZ06} due to its lack of convexity.

Our main contribution is a collection of results towards a general theory of variability measures, with particular emphasis on the three parametric classes mentioned above. Various novel properties are studied to underline the special role these measures play among other variability measures.  
Since statistical inference for VaR, ES, and expectiles is well developed (see e.g.~\cite{SW09} for VaR and \cite{KZ17} for the expectiles), the estimation of the corresponding variability measures is quite straightforward.  

The rest of the paper is organized as follows. In the remainder of this section, we introduce some notation. The definitions of the three classes of variability measures induced by VaR, ES, and the expectiles is presented in Section \ref{sec2}, with some basic properties. In Section \ref{sec3}, we summarize many   properties of some common variability measures which are arguably desirable in practice. A characterization result of these measures is established. The stochastic ordering of the three classes of variability measures based on pointwise comparison is discussed in \ref{sec4b}. 
In Section \ref{sec5}, we discuss non-parametric estimation of the three classes of variability measures. We
obtain the asymptotic normality and the asymptotic variances explicitly for the empirical estimators. It may be undesirable and financial unjustifiable to choose the same probability level for the three classes of variability measures induced by VaR, ES, and the expectiles; see \cite{LW19} for a detailed analysis on plausible equivalent probability levels when  ES is to replace  VaR. A simple analysis of a cross-comparison of an equivalent probability level for the variability measures using different distributions is carried out in Section \ref{sec6}. 
A small empirical analysis using the variability measures on the S\&P $500$ index is conducted in Section \ref{sec7}, where we observe the differences between these variability measures during different economic regimes. Further, we explore the symmetric variability orders between log-returns of Facebook and Berkshire Hathaway in 2020. In Section \ref{sec8}, we conclude the paper with some discussions on the suitability of the three classes in different situations. Appendix \ref{app:a1} contains a list of classic variability measures, and proofs of all results are put in Appendix \ref{app:B}.  


\textbf{Notation.} Throughout the paper,   $L^q$ is the set of all random variables in an atomless probability space $(\Omega, \mathcal A, \p)$ with finite $q$-th moment, $q\in (0,\infty)$, and   $L^\infty$ is the set of essentially bounded random variables.
$\X=L^0$ is the set of all random variables,
and $\mathcal M$ is the set of all distributions on $\R$. 
For any $X\in L^0$, $F_X$ represents the distribution function of $X$, $F^{-1}_X$ its left-quantile function, and $U_X$ is a uniform random variable such that $F^{-1}_X(U_X)=X$ almost surely. The existence of such a $U_X$ for any $X$ is given, for example, in Lemma A.32 of \cite{FS16}.
Two random variables $X$ and $Y$ are said to be comonotonic if there exist two increasing functions $f, g : \R \to \R$ such that $X=f(X+Y)$ and $Y=g(X+Y)$.
We write $X\laweq Y$  if  $X$ and $Y$ have the same distribution.  In this paper, the terms ``increasing" and ``decreasing" are meant in the non-strict sense. 

\section{Definitions}\label{sec2}
\subsection{Basic requirements for variability measures}

Generally speaking, a {variability measure} is a functional $\nu:\X \to [0,\infty]$ that quantifies the magnitude of variability of random variables.
In order for our definition to be as general as possible, we only require three natural properties.

\begin{definition}\label{def:1}
A \emph{variability measure} is a functional $\nu:\X \to [0,\infty]$ satisfying the following properties. 
\begin{enumerate}[(A1)]
\item[(A1)] \textit{Law invariance}: if $X,Y\in \X$ and $X\laweq Y$, then $\nu(X)=\nu(Y)$.
\item[(A2)] \textit{Standardization}: $\nu(m)=0$ for all $m\in \R$.
\item[(A3)] \textit{Positive homogeneity}: there exists $\alpha\in [0,\infty)$ such that $\nu(\lambda X)=\lambda^\alpha \nu(X)$ for any $\lambda> 0$ and $X\in \X$. The number $\alpha$ is called the \emph{homogeneity index} of $\nu$.
\end{enumerate}
\end{definition}

The three properties in Definition \ref{def:1} are the most basic, and they are satisfied by virtually all examples in the literature;  a useful variability measure typically satisfies other desirable properties (see Section \ref{sec3}). 
 
Some examples of classic variability measures are given in Appendix \ref{app:a1}. Notice that in the literature there are some relative measures of variability that are only defined for positive random variables, such as the Gini coefficient or the relative deviation (see Appendix \ref{app:a1}). 
In this paper, we do not deal with these cases, although our definition can be easily amended to include them by replacing $\X$ with a positive convex cone. 
We call the set $\X_\nu = \{X\in \X: \nu (X)<\infty\}$ the effective domain of $\nu$.
\begin{remark}
A deviation measure in the sense of \cite{RUZ06} satisfies, in addition to (A2) and (A3) with homogeneity index $1$, also subadditivity and strict positivity for non-constant random variables. As we will see in Section \ref{sec3},
the latter two properties are not satisfied by the inter-quantile difference. For this reason, our more general definition is more suitable here than the one of \cite{RUZ06}. 
 Alternatively, \cite{FWZ17} required location-invariance instead of positive homogeneity, but this property is not satisfied by relative variability measures.  Thus, we identify (A1), (A2), (A3) as the defining properties of a variability measure, and all other properties, such as location invariance and subadditivity, will be additional properties that may or may not be satisfied, as we will discuss in see Section \ref{sec3}.
\end{remark}

\begin{remark}
In applications, we may  choose the domain $\mathcal X$ of a variability measure as a convex cone contained in $L^0$.  For risk measures, the domain plays an essential role, which is often chosen as a general convex cone containing $L^\infty$, because many risk measures cannot be naturally extended to $L^0$; see e.g., \cite{FS12}.
 For variability measures defined on a convex cone $\mathcal X\subset L^0$, since it takes non-negative values (thus, no issues with $\infty-\infty$  which occur for some risk measures), we could always extend the domain by mapping $L^0\setminus \mathcal{X}$ to  $\{\infty\}$ without affecting the properties studied in this paper.
\end{remark}

\subsection{Three one-parameter families of risk measures}\label{sec:22}
Value at Risk (VaR), Expected Shortfall (ES) and expectiles are very popular financial risk measures (see e.g.~\cite{EPRWB14} and \cite{EKT15}).
We recall the basic definitions below. 
\begin{enumerate}[(i)]
\item The right-VaR  (right-quantile): for $p\in (0,1)$, $$Q_p(X)=\inf\{x\in \R: \p(X\le x) > p\}, ~~~X\in \X.$$
 The left-VaR (left-quantile):  for $p\in (0,1)$,  $$Q_p^-(X)=\inf\{x\in \R: \p(X\le x) \ge  p\}, ~~~X\in \X.$$
\item The ES:  for $p\in (0,1)$,  $$\ES_p(X)=\frac{1}{1-p}\int_p^1 Q_r(X) \d r,~~~X\in \X.$$
 The left-ES:  for $p\in (0,1)$,  $$\ES^-_p(X)=\frac{1}{p} \int_0^p Q_r(X) \d r ,~~~X\in \X.$$
\item The expectile:  for $p\in (0,1)$,  $$\ex_p(X)=\min \{x\in \R: p \E[(X-x)_+] \le (1-p)\E[(X-x)_-]\} ,~~~X\in L^1.$$
\end{enumerate}
In the above, $Q_p$ and $Q^{-}_p$ are finite on $L^0$,
while $\ES_p$,  $\ES_p^{-}$ and $\ex_p$ are finite on $L^1$. 
We only define expectiles on $L^1$ since generalizing them beyond $L^1$ is not natural; on the other hand, $\ES$ can be naturally defined on a set larger than $L^1$ by taking possibly infinite values.

\subsection{Three one-parameter families of variability measures}\label{sec:23}

We now introduce the variability measures induced by the aforementioned risk measures, that are the main object of the paper.

\begin{enumerate}[(i)]
\item The inter-quantile difference:  for $p\in [1/2,1)$,
$$
\Delta^Q_p(X) =  Q_p(X)-Q_{1-p}^{-}(X), ~~~X\in \X.
$$
It is obvious that $\Delta^Q_p$ is finite on $\X=L^0$.
\item The inter-ES difference:  for $p\in (0,1)$,
$$
\Delta^{\ES}_p(X) = \ES_p(X) -\ES_{1-p}^{-}(X), ~~~X\in \X.
$$
Here,   $\ES_p$ takes values in $(-\infty,\infty]$,
and $\ES_{1-p}^-$ takes values in $[-\infty,\infty)$, and hence 
the above $\Delta^{\ES}_p$ is well defined on $\X$.  
\item The inter-expectile difference: for $p\in (1/2,1)$,
$$
\Delta^{\ex}_p(X) =\ex_p(X) -\ex_{1-p}(X)  , ~~~X\in L^1,
$$
and we set by definition $\Delta^{\ex}_p(X)=\infty$ for $X\in \X\setminus L^1$.
\end{enumerate}
We consider also the limiting cases
$$ 
\Delta^Q_1(X)  = \Delta^{\ES}_1(X) = \Delta^{\ex}_1(X) =\esssup(X)-\essinf(X),~~~X\in \X,
$$
which is the range functional, and it is simply denoted by $\Delta_1$. 
Both $\Delta^Q_p$ and $\Delta^\ES_p$ belong to the class of  distortion riskmetrics (\cite{WWW20,WWW20b}), with many convenient theoretical properties. 
On the other hand,  $\Delta^\ex_p$ does not belong to this class, but it also has several nice properties, inherited from those of expectiles. 

In Theorems \ref{th:0}-\ref{th:1} and Table \ref{tab:1} below,   the range of $p$  is 
 $p\in [1/2,1)$ for $\Delta^Q_p$, $p\in (1/2,1)$ for $\Delta^{\ex}_p$,
 and $p\in (0,1)$ for $\Delta^{\ES}_p$.
\begin{theorem}\label{th:0} For each $p$, the following statements hold. 
\begin{enumerate}[(i)]
\item  $\Delta^Q_p$, $\Delta^{\ES}_p$, $\Delta^{\ex}_p$ and $\Delta_1$ are variability measures. 
\item  The effective domains of 
 $\Delta^Q_p$,    $\Delta^{\ES}_p$, $\Delta^{\ex}_p$  and $\Delta_1$ are 
 $L^0$, $L^1$, $L^1$, and $L^\infty$, respectively.
 \item Each of $\Delta^Q_p$, $\Delta^{\ES}_p$ and $\Delta^{\ex}_p$ is increasing in $p$. 
 \item For each $X\in \X$, the following alternative formulations hold: 
\begin{align*} 
\Delta^Q_p(X)& = Q_p(X) +Q_p(-X),  
\\\Delta^{\ES}_p(X) &= \ES_p(X) +\ES_p(-X), 
 \\ \Delta^{\ex}_p(X) &= \ex_p(X) +\ex_p(-X).
\end{align*}
\end{enumerate}
\end{theorem}

It is straightforward to check that for $p=1/2$, $\Delta_{p}^{\ES}$ is equal to two times the mean median deviation (see Appendix \ref{app:a1}, item (v)).  
The next proposition shows that it suffices to consider $p\in [1/2,1)$, as we will tacitly assume in most results of the next sections. 
\begin{proposition}\label{prop:eshalf}
For each $p\in (0,1)$,  $(1-p) \Delta_p^{\ES}= p \Delta_{1-p}^{\ES}$, 
and $\Delta_p^{\ES}=\frac{1}{1-p}\int_p^1 \Delta_q^Q \d q.$  \end{proposition}

\section{Comparative properties and characterization}\label{sec3}

In this section, we study comparative advantages of   $\Delta^Q_p$, $\Delta^{\ES}_p$ and $\Delta^{\ex}_p$, among with several other measures of variability, namely the standard deviation (STD), the variance,  the mean absolute deviation (MAD),  the Gini deviation (Gini-D),   and the range; see   Appendix \ref{app:a1} for the definition of these classic variability measures. 

We consider the following additional properties of a variability measure $\nu$, which are all arguably desirable in some situations. In what follows, $\leq_{\rm cx}$ is the convex order, defined by $X\leq_{\rm cx} Y$ if $\E[\phi(X)]\le \E[\phi(Y)]$ for all convex $\phi:\R\rightarrow \R$ such that the above two expectations exist.
\begin{enumerate}[(B1)]
\item  \emph{Relevance}: $\nu(X)>0$ if $X$ is not a constant, and there exists  $\beta\in (0,\infty)$ such that $\nu(X)\le \beta $ for all $X\in \X$ with $|X|\le 1$. 
 \item \emph{Continuity}: $\nu((X\wedge M )\vee (-M))\to \nu(X) $ as $M\to \infty$ for all $X\in \X$. 
\item \emph{Symmetry}: $\nu(X)=\nu(-X)$ for all $X\in \X$. 
\item \emph{Comonotonic additivity (C-additivity)}: $\nu( X+ Y)  = \nu(X) + \nu(Y)$ for all comonotonic $X,Y\in \X$.
\item \emph{Convex order consistency  (Cx-consistency)}: $\nu(X)\le \nu(Y)$ if $X\leq_{\rm cx} Y$. 
\item \emph{Convexity}: $\nu(\lambda X+(1-\lambda)Y) \le \lambda \nu(X) + (1-\lambda)\nu(Y)$ for all $X,Y\in \X$ and $\lambda \in [0,1]$. 
\item \emph{Mixture concavity (M-concavity)}: $\hat \nu$ is concave,
where   $\hat \nu:\mathcal M\to [0,\infty]$
is defined by $\hat \nu(F)=\nu(X)$ for $X\sim F$. 
\item \emph{Location invariance (L-invariance)}: $\nu(X+c)=\nu(X)$ for all $X\in \X$ and $c\in \R$.
\end{enumerate}
Relevance (B1) requires $\nu $  to report a positive value for all non degenerate distributions, 
and the value of $\nu(X)$ should not explode if $|X|\le 1$. 
Continuity (B2)  is very weak and unspecific to the effective domain of $\nu$. If $\nu$ is finite on $L^q$ for some $q\ge 1$, then (B2) is implied by   $L^q$ continuity.
Symmetry (B3) means that $\nu$ is indifferent to the positive and the negative sides of the distribution, and this property is in sharp contrast to  classic risk measures for which positive and negative values are interpreted very differently (deficit/surplus or loss/profit).  
The symmetry property of the measures of variability motivates and simplifies the application of the measures. 
C-additivity (B4) is a convenient functional property which allows for a characterization result below. The properties (B5)-(B7) describe natural requirements for $\nu$ to increase when the underlying distribution is more spread out in some sense; see \cite{WWW20} for further motivation and explanations of these properties.
Finally, (B8) requires that variability is measured independently of the location of the distribution and is indeed imposed as an axiom for measures of variability by \cite{FWZ17}.

In Table \ref{tab:1} below,  $\alpha$ represents the homogeneity index.
 Table \ref{tab:1} shows properties of different variability measures including the inter-quantile, inter-ES, and inter-expectile differences, as well as the aforementioned classic variability measures.

\begin{table}[htbp]
\begin{center} 
\small
\setstretch{1.3}
\setlength{\tabcolsep}{4pt}
\begin{tabular}{C{3.2cm}|C{1.1cm}|C{1.1cm}|C{1.1cm}|C{1.25cm}|C{1.1cm}|C{1.1cm}|C{1.25cm}|C{1.1cm}}
 
                                     & $\Delta_p^{Q}$             & $\Delta_p^{\ES}$           & $\Delta_p^{\ex}$           & variance                   & STD                        & MAD                        & Gini-D 
                                     &range 
                                     \\ \hline
 {relevance}      &  {NO}    &  {YES}   &  {YES}   &  {YES}   &  {YES}   &  {YES}   & YES  &  YES  \\ \hline 
 {continuity}       &  {YES}   &  {YES}   &  {YES}   &  {YES}   &  {YES}   &  {YES}   & YES  &  YES \\ \hline
 {symmetry}       &  {YES}   &  {YES}   &  {YES}   &  {YES}   &  {YES}   &  {YES}   & YES  & YES  \\ \hline  
 {C-additivity}       &  {YES}   &  {YES}   &  {NO}   &  {NO}   &  {NO}   &  {NO}   & YES  & YES \\ \hline
  {Cx-consistency} &  {NO}    &  {YES}   &  {YES}   &  {YES}   &  {YES}   &  {YES}   & YES  & YES  \\ \hline
 {convexity}       &  {NO}   &  {YES}   &  {YES}   &  {YES}   &  {YES}   &  {YES}   & YES  &  YES \\ \hline 
 {M-concavity}    &  {NO}    &  {YES}   &  {NO}    &  {YES}   &  {YES}   &  {NO }   & YES   &  YES \\ \hline
 {L-invariance}       &  {YES}   &  {YES}   &  {YES}   &  {YES}   &  {YES}   &  {YES}   & YES  & YES \\ \hline
 {homogeneity ($\alpha$)}             &  {1}     &  {1}     &  {1}     &  {2}     &  {1}     &  {1}     & 1  &  1  \\ \hline
 {effective domain}         &  {$L^0$} &  {$L^1$} &  {$L^1$} &  {$L^2$} &  {$L^2$} &  {$L^1$} & $L^1$ & $L^\infty$ \\  \hline
\end{tabular}
\end{center}
\caption{Properties of variability measures.}\label{tab:1}
\end{table}

\begin{theorem}\label{th:1}
The statements in  Table \ref{tab:1} hold true. 
\end{theorem}

The proof of Theorem \ref{th:1}, thus checking the properties in Table \ref{tab:1}, relies on several existing results  on properties of risk measures and distortion riskmetrics from \cite{NP87}, \cite{BKMR14, BBP18}, \cite{LCLW20} and \cite{WWW20}.


Notably, the inter-ES difference satisfies all properties (B1)-(B8), along with the Gini deviation  and the range. 
Next,
we establish that any
variability measure  satisfying (B1)-(B8)   admits a representation as a mixture of $\Delta_p^\ES$ for $p\in (0,1]$.

\begin{theorem}\label{th:2}
The following statements are equivalent for a variability measure $\nu:\X\to [0,\infty]$:
\begin{enumerate}[(i)]
\item $\nu$ satisfies (B1)-(B8).
\item $\nu$ satisfies (B1)-(B4) and one of (B5)-(B6).
\item $\nu$ is a   mixture of $\Delta^{\ES}_p$, that is, there exists a  finite Borel measure $\mu \ne 0$ on $(0,1]$ such that 
 \begin{align}\label{eq:repES}
\nu(X) = \int_0^1\Delta^{\ES}_p (X)\d \mu(p),~~X\in \X.
\end{align}
\end{enumerate}

\end{theorem}

The measure $\mu$ in \eqref{eq:repES} for a given $\nu$ is generally not unique. 
Using Proposition \ref{prop:eshalf}, we can require $\mu$ in \eqref{eq:repES} to be   supported on $[1/2,1]$ instead of $(0,1]$.

\begin{example}\label{ex:2}
There are three variability measures in Table \ref{tab:1} that satisfy all of (B1)-(B8), and each admit a representation as in Theorem \ref{th:2}. We give below a corresponding measure $\mu$ for each of them.
\begin{enumerate}
\item $\Delta_p^\ES$ for $p\in (0,1)$: $\mu=\delta_p$.  
\item The Gini deviation:   $  \mu(\d x)   = (1-x) \d x $ on $[0,1]$.  
\item The range $\Delta_1$: $\mu=\delta_1$.
\end{enumerate}

\end{example}
As we have seen from Theorem \ref{th:1}, all of  $\Delta_p^Q$, 
$\Delta_p^\ES$, $\Delta_p^\ex$ are invariant under location shifts.
In the next result, we show that each of the one-parameter families  $\Delta_p^Q$, 
$\Delta_p^\ES$, $\Delta_p^\ex$ characterize a symmetric distribution up to location shifts. 
\begin{proposition}\label{th:4}
Suppose that $X$ has a symmetric distribution, i.e., $X \laweq -X$.  
Each of the curves 
$
p\mapsto \Delta_p^Q (X)
$,
$
p\mapsto \Delta_p^\ES (X)
$ and
$
p\mapsto \Delta_p^\ex (X)
$  for $p\in (1/2,1)$, if it is finite, 
 uniquely determines the distribution of $X$. 
\end{proposition}
\begin{remark}
If the distribution of $X$ is not symmetric, none of $
p\mapsto \Delta_p^Q (X)
$,
$
p\mapsto \Delta_p^\ES (X)
$ and
$
p\mapsto \Delta_p^\ex (X)
$  for $p\in (1/2,1)$
determines its distribution up to location shifts.  
This is because the inter-quantile difference curve $p\mapsto Q_p-Q^-_{1-p}$  does not determine the quantile curve $p\mapsto Q_p$.
For instance, 
given a  quantile curve $p\mapsto  Q_p(X)$, we can define another  quantile curve  $p\mapsto  Q_p(Y)$  by
$$
Q_p(Y) = Q_p(X) + f(p),~~p\in (0,1),
$$
where $f(p)$ is any  continuous function satisfying $f(p)=f(1-p)$ for $p\in (0,1/2)$,
such that $Q_p(Y)$ is increasing in $p$. 
The inter-quantile difference curves of $X$ and $Y$ are the same, but the distributions of $X$ and $Y$ are not the same up to a location shift unless $f$ is a constant. 
\end{remark}
\begin{remark}
From \cite{K01} it is well known that any coherent risk measure admits a representation as a supremum of mixtures of ES; see \cite{BKMR14} for the case of expectiles. One naturally wonders whether an inter-expectile difference can be represented as the supremum of mixtures of inter-ES differences, i.e., the supremum over functions of form \eqref{eq:repES}. Rather surprisingly, it turns out that such a relationship does not hold in general, as illustrated by Example \ref{ex} in Section 4.
\end{remark}

\section{Symmetric variability orders} \label{sec4b}

Since the variability measures can be easily estimated from real data (see Section \ref{sec5} below), one may conclude some ordering relations between two data sets with ordered measures of variability. For this purpose, we consider stochastic orders induced by pointwise comparison of inter-quantile, inter-ES, and inter-expectile differences. The first case has been studied in \cite{TC05} under the name of quantile spread order, defined as follows: 
$$
X \leq_{\rm QS} Y 
\text{ if } \Delta_p^Q(X) \leq \Delta_p^Q(Y), \text { for each } p \in (1/2,1).
$$
Note that the order $\leq_{\rm QS}$ is weaker than the well-known dispersive order, defined by 
$$
X \leq_{\rm disp} Y 
\text{ if } Q_\beta(X)-Q_\alpha(X) \leq Q_\beta(Y)-Q_\alpha(Y), \text { for each } 0<\alpha < \beta <1,
$$
for which we refer e.g.~to \cite{MS02} and \cite{SS07}. 
We define two stochastic orders based on inter-ES and inter-expectile differences as follows:
\begin{align*}
X \leq_{\rm \Delta{\text -}\ES} Y 
&\text{ if } \Delta_p^{\ES}(X) \leq \Delta_p^{\ES}(Y), \text { for each } p \in (1/2,1), \\
X \leq_{\rm \Delta{\text -}\ex} Y 
&\text{ if } \Delta_p^{\ex}(X) \leq \Delta_p^{\ex}(Y), \text { for each } p \in (1/2,1).
\end{align*}
It turns out that for symmetric random variables, these orders are equivalent to the dilation order $\leq_{\rm dil}$, defined by 
$$
X \leq_{\rm dil} Y \text { if } X-\E[X] \leq_{\rm cx} Y-\E[Y],  
$$
as shown in (v) and (vi) below; the other properties are summarized in the following. 
\begin{proposition}\label{th:orders} Let $X, Y \in L^1$. The following statements hold: 
\begin{enumerate}[(i)]
\item  For any $c \in \R$, $X \leq_{\rm \Delta{\text -}\ES} Y \iff X + c \leq_{\rm \Delta{\text -}\ES} Y$;  $X \leq_{\rm \Delta{\text -}\ex} Y \iff X + c \leq_{\rm \Delta{\text -}\ex} Y$.
\item If $|a| \geq 1$, then $X \leq_{\rm \Delta{\text -}\ES} aX$ and $X \leq_{\rm \Delta{\text -}\ex} aX$;
\item $X \leq_{\rm \Delta{\text -}\ES} Y \iff X \leq_{\rm \Delta{\text -}\ES} -Y$; $X \leq_{\rm \Delta{\text -}\ex} Y \iff X \leq_{\rm \Delta{\text -}\ex} -Y$.
\item $X \leq_{\rm QS} Y \Longrightarrow X \leq_{\rm \Delta{\text -}\ES} Y$.
\item $X \leq_{\rm dil} Y \Longrightarrow X \leq_{\rm \Delta{\text -}\ES} Y$ and $X \leq_{\rm \Delta{\text -}\ex} Y$.
\item  If $X$ and $Y$ are symmetric with respect to their means, then 
$X \leq_{\rm dil} Y \Longleftrightarrow {X \leq_{\rm \Delta{\text -}\ES} Y} $ $ \Longleftrightarrow X \leq_{\rm \Delta{\text -}\ex} Y$.
\end{enumerate}
\end{proposition}
In case $X$ or $Y$ is not symmetric, then the equivalence relations in (vi) may fail, as the following simple example shows.  Therefore, the two new orders  $\leq_{\rm \Delta{\text -}\ES}$ and $\leq_{\rm \Delta{\text -}\ex}$  are generally weaker than the dilation order. This provides  more flexibility for these new orders in real-data applications, as we will illustrate in Section \ref{sec7}.
\begin{example}
Let 
$$
X = 
\begin{cases}
-3/4 &\text{with prob. } 1/3\\
-1/2 &\text{with prob. } 1/3\\
\phantom{-}5/4 &\text{with prob. } 1/3\\
\end{cases}
~~~\mbox{and}~~~
Y = 
\begin{cases}
-1 &\text{with prob. } 1/2\\
\phantom{-}1 &\text{with prob. } 1/2.\\
\end{cases}
$$
Then $\E[X]=\E[Y]=0$, and 
\begin{align*}
\Delta_p^{\ES}(X) &=
\begin{cases}
2 & 2/3 \leq p \leq 1\\
\frac{2}{3(1-p)} & 1/2 \leq p \leq 2/3,
\end{cases}
~~~\mbox{and}~~~
\Delta_p^{\ES}(Y) =2, \, ~~1/2 \leq p \leq 1.
\end{align*}
Hence, $\Delta_p^{\ES}(X) \leq \Delta_p^{\ES}(Y)$ for each $p \in [1/2,1]$. Also, by a straightforward computation,
\begin{align*}
\ex_p(X) &= 
\begin{cases}
\frac{6p-3}{4+4p}  &\text { if } 0 \leq p \leq 1/8 \\
\frac{10p-5}{8-4p} &\text { if } 1/8 \leq p \leq 1, \\
\end{cases}
~~~\mbox{and}~~~
\ex_p(Y)  = 2p-1 \text{ for } 0 \leq p \leq 1;\\
\Delta_p^{\ex}(X) &=
\begin{cases}
\frac{30p-15}{4(2-p)(1+p)} &\text { if } 1/2 \leq p \leq 7/8 \\
\frac{4p-2}{2-p} &\text { if } 7/8 \leq p \leq 1,\\
\end{cases}
~~~\mbox{and}~~~
\Delta_p^{\ex}(Y)=4p-2 \text{ for } 1/2 \leq p \leq 1.
\end{align*}
It follows  that   $\Delta_p^{\ex}(X) \leq \Delta_p^{\ex}(Y)$ for each $p \in [1/2,1]$.
However, $X \not\leq_{\rm dil} Y$ because $X$ and $Y$ have the same mean, and the support of $X$ is not contained in that of $Y$.  This shows that $\leq_{\Delta{\text -}\ES}$ and $\leq_{\Delta{\text -}\ex}$ do not imply $\leq_{\rm dil}$.
\end{example}
Finally, in the asymmetric case the $\leq_{\Delta{\text -}\ES}$ and $\leq_{\Delta{\text -}\ex}$ orders are not related. In the next example we have that $X \leq_{\Delta{\text -}\ES} Y$ but $X \not \leq_{\Delta{\text -}\ex} Y$, and a (real-data) example in which the opposite situation occurs can be found in Section \ref{sec7}.
\begin{example}\label{ex}
Let
$$
X = 
\begin{cases}
-1 &\text{with prob. } 1/4\\
\phantom{-} 1 &\text{with prob. } 3/4\\
\end{cases}
~~~\mbox{and}~~~
Y = 
\begin{cases}
-1 &\text{with prob. } 1/4\\
\phantom{-}0 &\text{with prob. } 1/4\\
\phantom{-}1 &\text{with prob. } 1/2.\\
\end{cases}
$$
Then
\begin{align*}
\Delta_p^{\ES}(X) &=
\begin{cases}
2 & 3/4 \leq p \leq 1\\
\frac{1}{2(1-p)} & 1/2 \leq p \leq 3/4,
\end{cases}
~~~\mbox{and}~~~
\Delta_p^{\ES}(Y) &=
\begin{cases}
2 & 3/4 \leq p \leq 1\\
1+\frac{1}{4(1-p)} & 1/2 \leq p \leq 3/4.
\end{cases}
\end{align*}
Hence $\Delta_p^{\ES}(X) \leq \Delta_p^{\ES}(Y)$ for each $p \in [1/2,1]$ and $X \leq_{\Delta{\text -}\ES} Y$. Also, 
\begin{align*}
\ex_p(X)  &= \frac{4p-1}{1+2p}\text{ for } 0 \leq p \leq 1
~~~\mbox{and}~~~
\ex_p(Y) = 
\begin{cases}
\frac{3p-1}{2p+1}  &\text { if } 0 \leq p \leq 1/3 \\
\frac{3p-1}{2} &\text { if } 1/3 \leq p \leq 1, \\
\end{cases}
\\
\Delta_p^{\ex}(X) &=\frac{12p-6}{(1+2p)(3-2p)} \text{ for } 1/2 \leq p \leq 1
~~\mbox{and}~~
\Delta_p^{\ex}(Y) =
\begin{cases}
\frac{-6p^2+17p-7}{6-4p} &\text { if } 2/3 \leq p \leq 1\\
\frac{6p-3}{2} &\text { if } 1/2 \leq p \leq 2/3. \\
\end{cases}
\end{align*}
Since $\Delta_{2/3}^{\ex}(X)=\frac{18}{35}>\frac{1}{2}=\Delta_{2/3}^{\ex}(Y)$, it follows that 
$X \not \leq_{\Delta{\text -}\ex} Y$. 
\end{example}
\section{Non-parametric estimators}\label{sec5}

The properties of non-parametric estimators  of 
 $ \Delta_p^Q (X)
$,
$ \Delta_p^\ES (X)
$ and
$ \Delta_p^\ex (X)$
can be derived from those of  VaR, ES and expectiles, as we will explain in this section.

Suppose $X_1,X_2,\dots, X_n \in L^1 $ is an iid sample from a random variable $X$.
Recall that the empirical distribution $\widehat F_n$ of $X_1,\dots,X_n$ is given by 
$$
\widehat F_n(x) = \frac 1 n \sum_{j=1}^n \id_{ \{X_j\le x\}},~~x\in \R.
$$
Let $\widehat\Delta_p^Q(n)$ be the empirical estimator of $\Delta_p^Q(X)$, obtained by applying $\Delta_p^Q$ to the empirical distribution of $X_1,\dots,X_n$. 
Similarly, let $\widehat\Delta_p^\ES (n)$ and $\widehat\Delta_p^\ex (n)$ be the empirical estimators of  $\Delta_p^\ES (X)$ and $\Delta_p^\ex (X)$.
We will establish  consistency and  asymptotic normality of the empirical estimators, based on   corresponding results on empirical estimators of VaR, ES and expectiles in the literature, e.g., \cite{CT05}, \cite{C08}, and \cite{KZ17}. 
We make the following standard regularity assumption on the distribution of the random variable $X$.
\begin{enumerate}
\item[(R)]   The distribution $F$ of $X\in L^1$ is supported on a convex set and has a positive density function $f$ on the support.
\end{enumerate} 
Denote by $g=f\circ{F^{-1}}$ and let $p\in (1/2,1)$. We will show in the next theorem that the asymptotic variances of the empirical estimators for $\Delta^Q_p$ and $\Delta^\ES_p$ are given by, respectively,  \begin{align}
\sigma_Q^2&= \frac{  p(1-p)}{(g(p))^2} + \frac{  p(1-p)}{(g(1-p))^2}    -2 \frac{(1-p)^2}{g(p)g(1-p)},\label{eq:asymQ}\\  
\sigma_\ES^2& =   \frac{1}{(1-p)^2}  \left(\int_{[p,1]^2\cup [0,1-p]^2} -2 \int_{[p,1]\times [0,1-p]}  \right) \frac{s\wedge t-st}{g(t)g(s)}\d t \d s  , \label{eq:asymES}\end{align} 
and that for $\Delta^\ex_p$ is given by 
\begin{align}
 \sigma_\ex^2&  =  s_p^\ex + s_{1-p}^\ex - 2 c_p^{\ex}, \label{eq:asymex} 
\end{align} 
 where for $r\in\{p,1-p\}$,
      \begin{align*}
 f^\ex_{r,F}(t)&= \frac{ (1-r) \mathbb \id_{\{t\le \ex_r(X)\}}  +
r \mathbb \id_{\{t>\ex_r(X)\}}}{(1-2r)F(\ex_r(X))+ r},~~~~t\in \R,\\
    s_r^\ex &= \int_{-\infty}^\infty\int_{-\infty}^\infty f^\ex_{r,F}(t)f^\ex_{r,F}(s)F(t \wedge s)(1-F(t \vee s))\d t\d s,\\
c _r^{\ex} &= \int_{-\infty}^\infty\int_{-\infty}^\infty f^{\ex}_{r,F}(t)  f^{\ex}_{1-r,F}(s)F(t \wedge s)(1-F(t \vee s))\d t \d s.   \end{align*} 
\begin{theorem}\label{th:5}
Suppose that  $p\in (1/2,1)$ and Assumption (R) holds.
\begin{enumerate}[(i)] 
\item $ \widehat \Delta_p^Q (n)\pto \Delta_p^Q (X)
$,
$\widehat\Delta_p^\ES (n) \pto \Delta_p^\ES (X)
$ and
$\widehat\Delta_p^\ex (n)  \pto \Delta_p^\ex (X)$ 
as $n\to \infty$.
\item If $X\in L^{2+\delta}$ for some $\delta>0$, then
 \begin{align*}
  \sqrt{n} (\widehat \Delta_p^Q (n) - \Delta_p^Q (X))&  \dto \mathrm{N}(0,\sigma^2_Q),
\\  \sqrt{n} (\widehat \Delta_p^\ES (n) - \Delta_p^\ES (X))& \dto \mathrm{N}(0,\sigma^2_\ES),
 \\
  \sqrt{n} (\widehat \Delta_p^\ex (n) - \Delta_p^\ex (X)) &\dto \mathrm{N}(0,\sigma^2_\ex),\end{align*}
where $ \sigma^2_Q$, $\sigma^2_\ES$ and $\sigma^2_\ex$ are given in \eqref{eq:asymQ}, \eqref{eq:asymES} and \eqref{eq:asymex}, respectively.  
\end{enumerate}
\end{theorem}

Simulation results are presented in Figure \ref{fig:1} for $p=0.9$ in the case of standard normal and Pareto risks with tail index $4$, that confirm the asymptotic normality of  the empirical estimators in Theorem \ref{th:5}.  
More general asymptotic results for $\alpha$-mixing processes could be similarly established using results in \cite{C08} and \cite{KZ17}. For the sake of space we do not discuss here the case of dependent observations.

\begin{figure}[htbp]
	\centering
    \begin{subfigure}{.48\textwidth}
		\centering
		\includegraphics[scale=0.6]{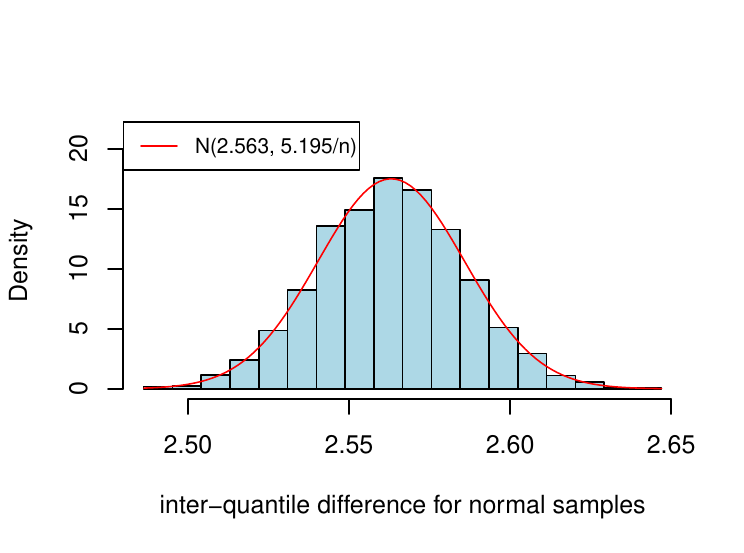}
		\caption{Histogram of $\widehat \Delta_p^Q (n)$ for $\mathrm N(0,1)$.}
		\label{fig1c}
	\end{subfigure}
	\begin{subfigure}{.48\textwidth}
		\centering 
		\includegraphics[scale=0.6]{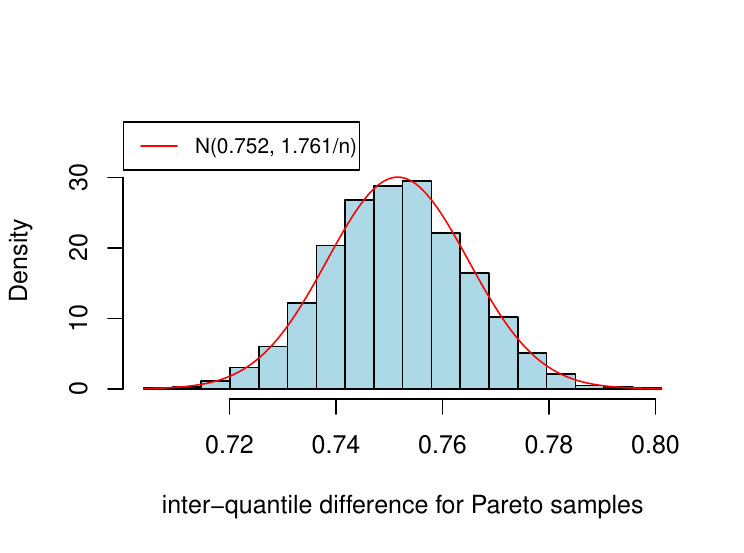}
		\caption{Histogram of $\widehat \Delta_p^Q (n)$ for Pareto($4$).}
		\label{fig1d}
	\end{subfigure}
	\centering
    \begin{subfigure}{.48\textwidth}
		\centering
		\includegraphics[scale=0.6]{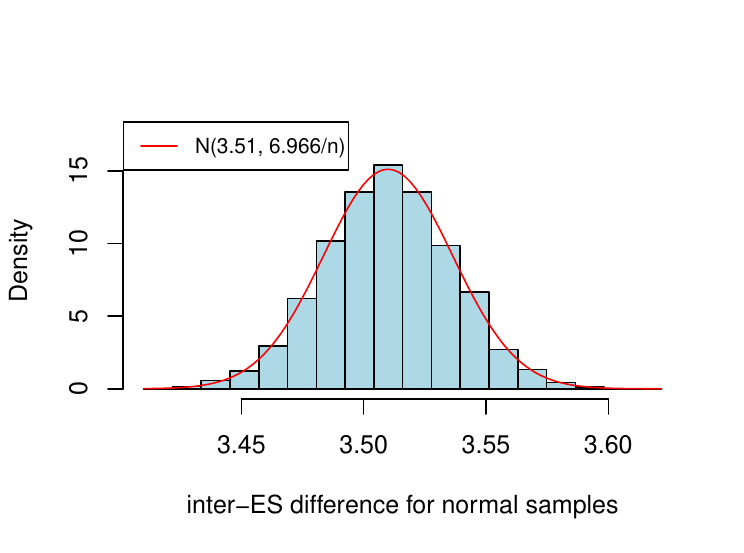}
		\caption{Histogram of $\widehat \Delta_p^{\ES} (n)$ for $\mathrm N(0,1)$.}
		\label{fig1c}
	\end{subfigure}
	\begin{subfigure}{.48\textwidth}
		\centering 
		\includegraphics[scale=0.6]{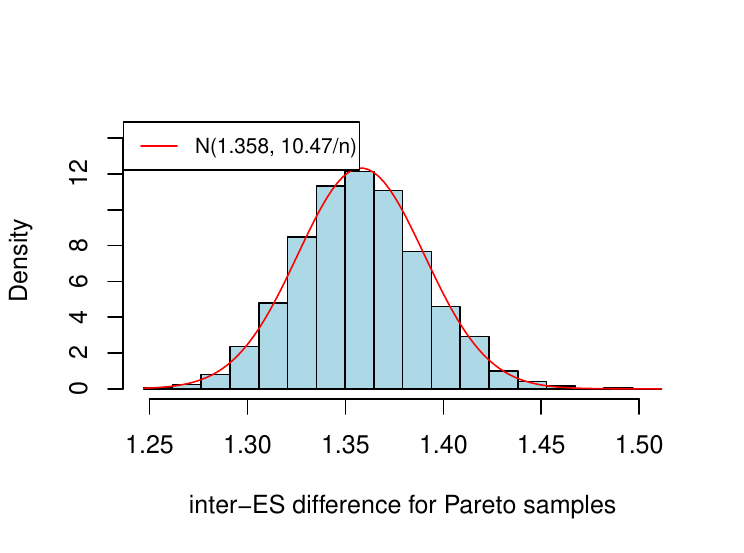}
		\caption{Histogram of $\widehat\Delta_p^{\ES}(n)$ for Pareto($4$).}
		\label{fig1d}
	\end{subfigure}
	\centering
    \begin{subfigure}{.48\textwidth}
		\centering
		\includegraphics[scale=0.6]{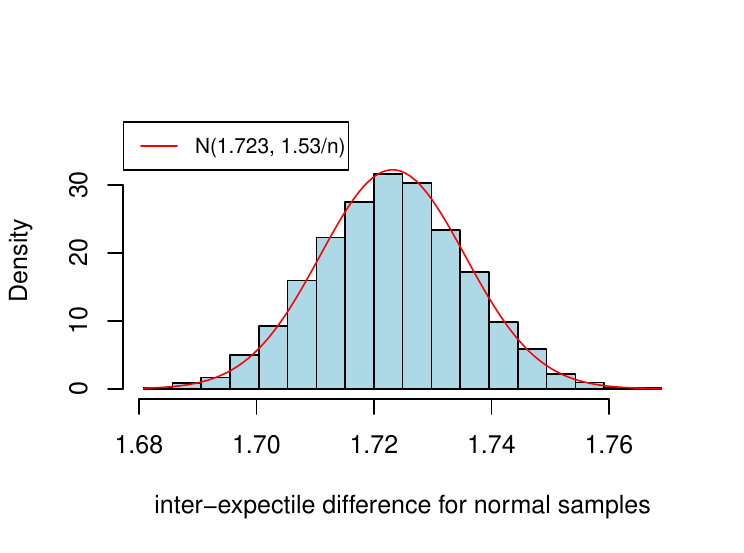}
		\caption{Histogram of $\widehat \Delta_p^{\rm ex} (n)$ for $\mathrm N(0,1)$.}
		\label{fig1c}
	\end{subfigure}
	\begin{subfigure}{.48\textwidth}
		\centering 
		\includegraphics[scale=0.6]{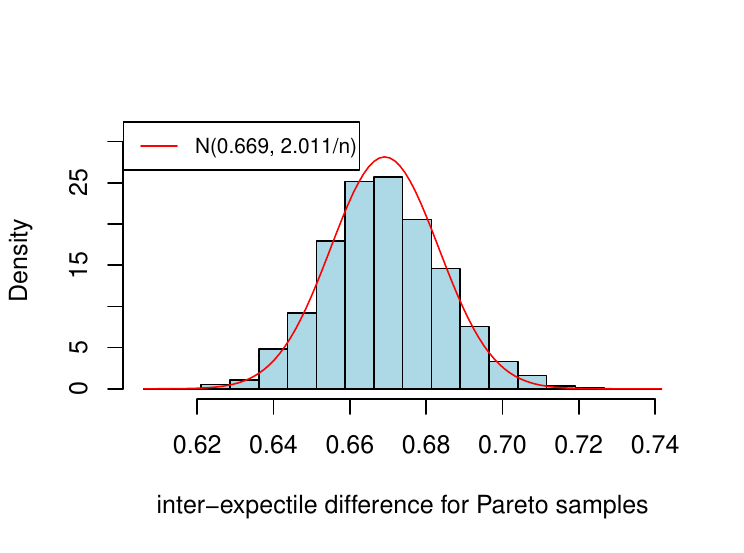}
		\caption{Histogram of $\widehat\Delta_p^{\rm ex}(n)$ for Pareto($4$).}
		\label{fig1d}
	\end{subfigure}
\caption{Histograms of empirical  estimators for simulated normal and Pareto risks, plotted against the  density of their asymptotic normal distributions in Theorem \ref{th:5} (with variance normalized by the sample size $n$). Each histogram is computed from 5,000 replications with sample size $n=10,000$. The parameter $p$ is set to $0.9$ in all simulation experiments.}	
\label{fig:1}
 \end{figure}
 
 \begin{remark}
For part (i) of Theorem \ref{th:5}, the assumption (R) is used to guarantee that the empirical quantiles converge to the true quantile (more precisely, we only need the quantile function to be continuous at $p$ and $1-p$); this is not needed for the consistency statements on $\Delta_p^{\rm ES}$ and $\Delta_p^{\rm ex}$ in part (i).  
 \end{remark}

  \begin{remark}
  For a convex risk measure  $\rho: H^{\Psi} \to \R$ on an Orlicz heart $H^\Psi$, 
  \citet[Theorem 2.6]{KSZ14} showed that the empirical estimator $\widehat \rho(n)$ is strongly consistent for a stationary and ergodic data sequence $(X_n)_{n\in \N}$; that is, $\widehat \rho(n) \to \rho(X_1)$ almost surely.
  For a general variability measure $\nu$, 
a similar result holds if we further assume that $\nu$ is norm-continuous on $H^{\Psi}$, following the same arguments as in the proof of Theorem 2.6 of \cite{KSZ14}, where the only used property of $\rho$ is norm-continuity.
This statement includes consistency of $\widehat\Delta_p^\ES (n)$ and $\widehat\Delta_p^\ex (n)$ on $L^1$ in part (i) of Theorem \ref{th:5} as special cases.
 For real-valued convex risk measures, norm-continuity is implied by monotonicity and convexity. 
To establish such continuity, 
monotonicity is essential and it plays the role of positivity in the Namioka-Klee theorem for linear functions; see \cite{BF09}. For convex variability measures, since monotonicity is not satisfied, we have to assume norm-continuity for the analog of Theorem 2.6 of \cite{KSZ14}. 
 \end{remark}

\section{A rule of thumb for cross comparison}
\label{sec6}

As mentioned in the introduction, we are interested in comparing the inter-quantile, the inter-ES and the inter-expectile differences. 
Due to the different meanings of the parameter $p$ in $\VaR_p$, $\ES_p$ and $\ex_p$, 
there is no reason to directly compare $\Delta^Q_p$, $\Delta^\ES_p$ and $\Delta^\ex_p$ using the same probability level $p$.
For a fair cross comparison, 
we may calibrate $p,q,r$ such that the variability measures  have the same value,   that is,
$$\Delta^Q_{p} = \Delta_{q}^\ES = \Delta^\ex_{r},$$
 for some common choices of distributions. In particular, we will consider normal (N), $t$- and   exponential
distributions as benchmarks, and the curves of $q$ and $r$ in terms of $p$ for these distributions are plotted in Figure \ref{fig:2}. 
We observe that the values of $r$ is typically much closer to $1$ than the corresponding $p$ or $q$.
The matching value of $q$ is smaller than the corresponding  $p$ but the relationship between $q$ and $p$ is close to linear; a corresponding observation on comparing VaR and ES is noted by \cite{LW19}, where they obtained the ratios $(1-q)/(1-p)\approx 2.5  $   for normal risks and $(1-q)/(1-p)=e\approx 2.72$ for exponential risks (this corresponds to the straight line in Figure \ref{fig1d}).

In empirical studies, 
it has been costumary in the literature to use the matching values for normal distribution as a rule of thumb for general comparisons; note that the location and scale parameters are irrelevant for such a comparison due to location-invariance and positive homogeneity. 
Roughly, we obtain 
$$\Delta^Q_{p} \approx \Delta_{q}^\ES \approx \Delta^\ex_{r}$$
for $(p,q,r)\in\{ (0.9, 0.75, 0.97), (0.95,0.875,0.99), (0.99,0.97,0.999)\}$.
For the particular choice of $p=0.95$,  it means that  $\Delta^Q_{0.95} \approx \Delta^{\ES}_{0.875} \approx \Delta^\ex_{0.99}$ for normal risks. We will compare these variability measures on real data in the next section, confirming the well-known departure from normality of financial returns.

\begin{figure}[htbp]
	\centering
    \begin{subfigure}{.4\textwidth}
		\centering
		\includegraphics[scale=0.6]{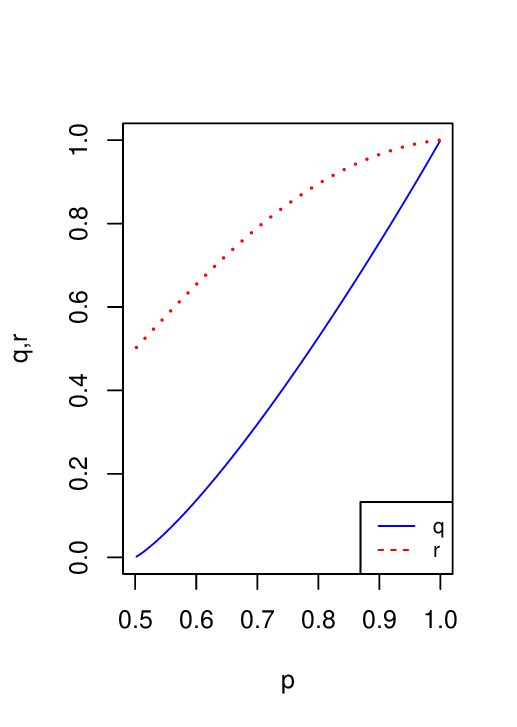}
		\caption{$q,r$ for $p\in (0.5,1)$ in $\mathrm N(0,1)$.}
		\label{fig1c}
	\end{subfigure}
	\begin{subfigure}{.4\textwidth}
		\centering 
		\includegraphics[scale=0.6]{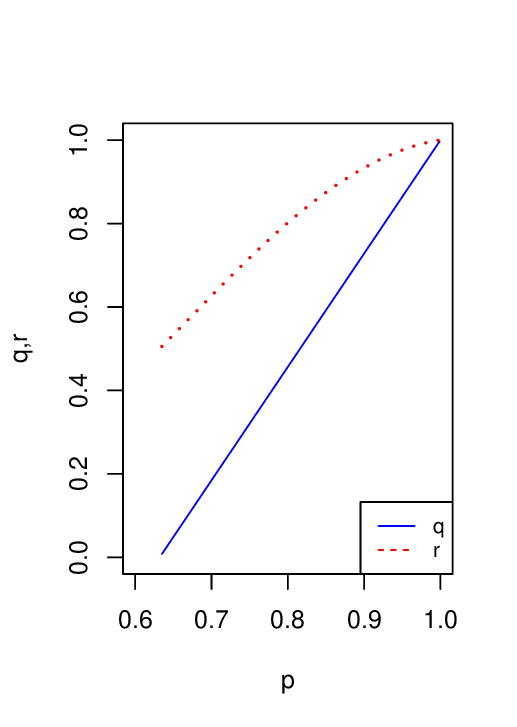}
		\caption{$q,r$ for $p\in (0.5,1)$ in $\exp(1)$.}
		\label{fig1d}
	\end{subfigure} 
     \begin{subfigure}{.4\textwidth}
		\centering
		\includegraphics[scale=0.6]{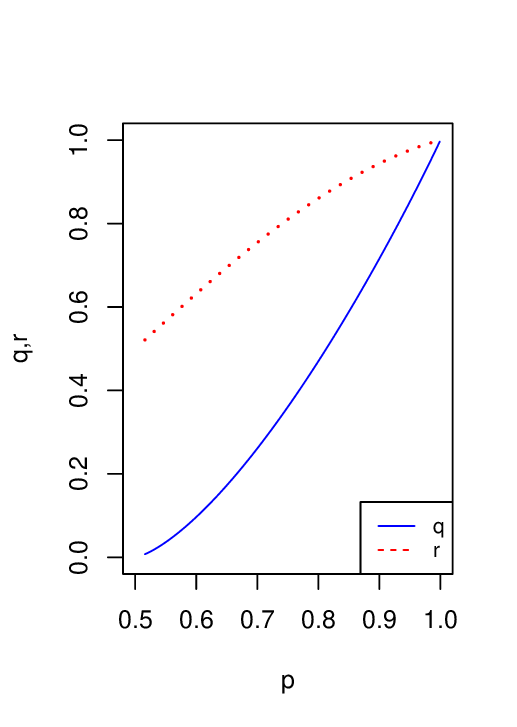}
		\caption{$q,r$ for $p\in (0.5,1)$ in $t(4)$.}
		\label{fig1e}
	\end{subfigure}
	 \begin{subfigure}{.4\textwidth}
		\centering
		\includegraphics[scale=0.6]{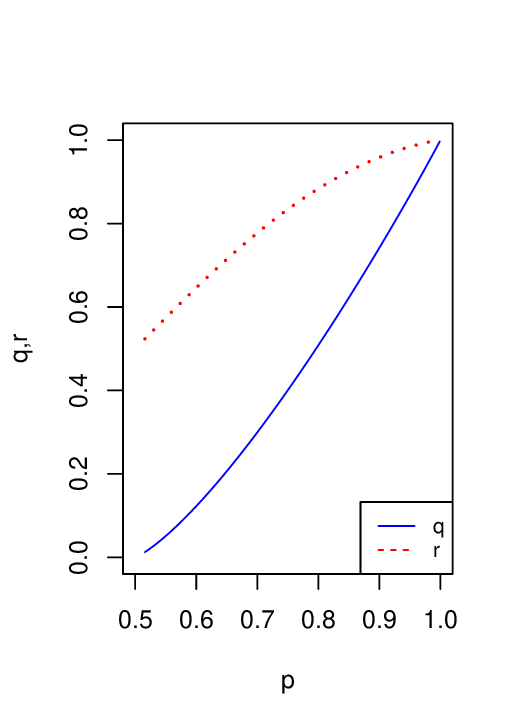}
		\caption{$q,r$ for $p\in (0.5,1)$ in $t(10)$.}
		\label{fig1f}
	\end{subfigure}
\caption{$q,r$ such that $\Delta^Q_{p} = \Delta_{q}^\ES = \Delta^\ex_{r}$ for $p\in (0.5,1)$.}	
\label{fig:2}
 \end{figure}

\section{Empirical analysis}
\label{sec7}

In this section, we illustrate   the three classes of variability measures studied in this paper  by means of a few empirical studies on financial data. 

We first analyze the difference between the performances of these variability measures during different periods of time (different economic regimes). Our data are the historical price movements
spanned from $01/04/1999$ to $06/30/2020$ of the S\&P $500$ index.\footnote{The source of the price data is Yahoo Finance.} We use its daily log-loss data\footnote{We use the log-loss (negative log-return) to be consistent  with most studies on financial asset return data. Note that since our  variability measures are symmetric, using log-losses is equivalent to using log-returns.} over the observation period with moving window of $253$ days for daily estimation of the variability measures. 
To compare the relative performance of the three measures, 
we report the   ratios  $\Delta_q^\ES/\Delta_r^\ex$ and $\Delta_q^\ES/\Delta_p^Q$ for the S\&P $500$ daily log-losses  in   Figures \ref{ratio2} and  \ref{ratio1} using the rule of thumb for $(p,q,r)$ obtained in Section \ref{sec6} induced by the normal distribution. In Figure \ref{ratio2}, spikes in the ratio of $\Delta_q^\ES/\Delta_p^Q$ are located around the $2008$ subprime crisis and the COVID-$19$ period. On the other hand, the ratio $\Delta_q^\ES/\Delta_r^{\ex}$ in Figure \ref{ratio1} experiences a down-slide around the subprime crisis and the COVID-$19$ period. 
These results suggest that $\Delta_q^\ES$ is more sensitive to extremely large losses than $\Delta_p^Q$, but $\Delta_r^\ex$ is even  more sensitive than $\Delta_q^\ES$.
Recall that these ratios should be $1$ if the underlying losses are normally distributed,
whereas we observe  $\Delta_q^\ES/\Delta_p^Q>1$ and $\Delta^\ES_q/\Delta_r^{\ex}<1$ for most dates during the period of 2000 - 2020 ($\Delta_q^\ES/\Delta_r^{\ex}$ is almost always smaller than 1). 
Hence, Figures \ref{ratio2} and  \ref{ratio1} confirm that the log-losses of S\&P $500$  are not normally distributed, and in fact, they typically show paretian tails, as is well studied in the literature (see, e.g., \cite{MFE15}).
 
\begin{figure}[H]
	\centering
     \begin{subfigure}{.4\textwidth}
		\centering
		\includegraphics[scale=0.4]{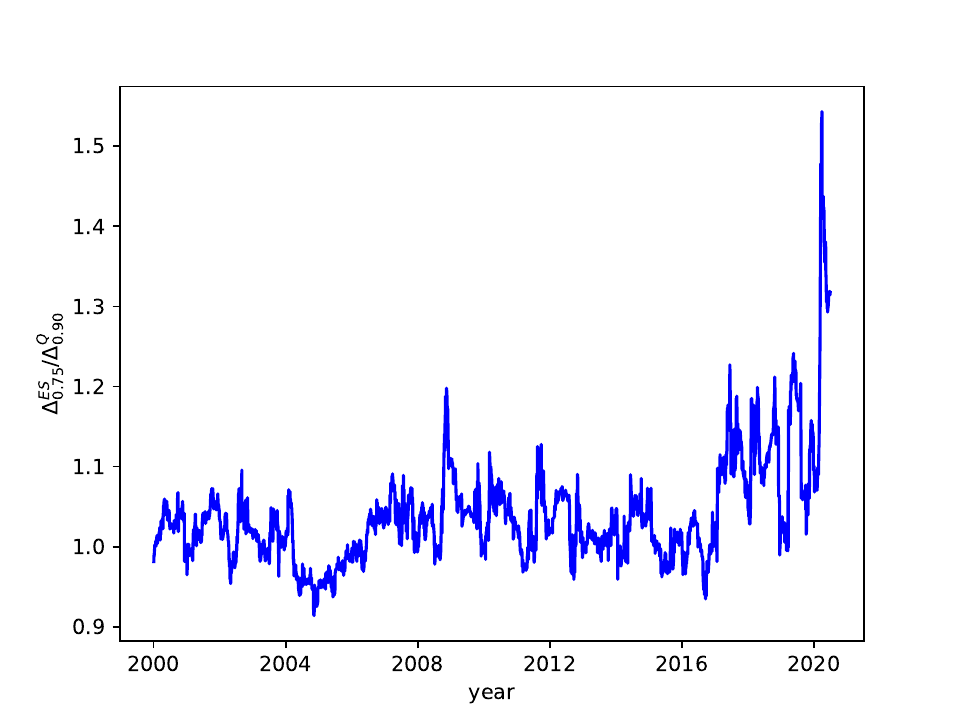}
		\label{fig1e}
	\end{subfigure}
    \begin{subfigure}{.4\textwidth}
		\centering 
		\includegraphics[scale=0.4]{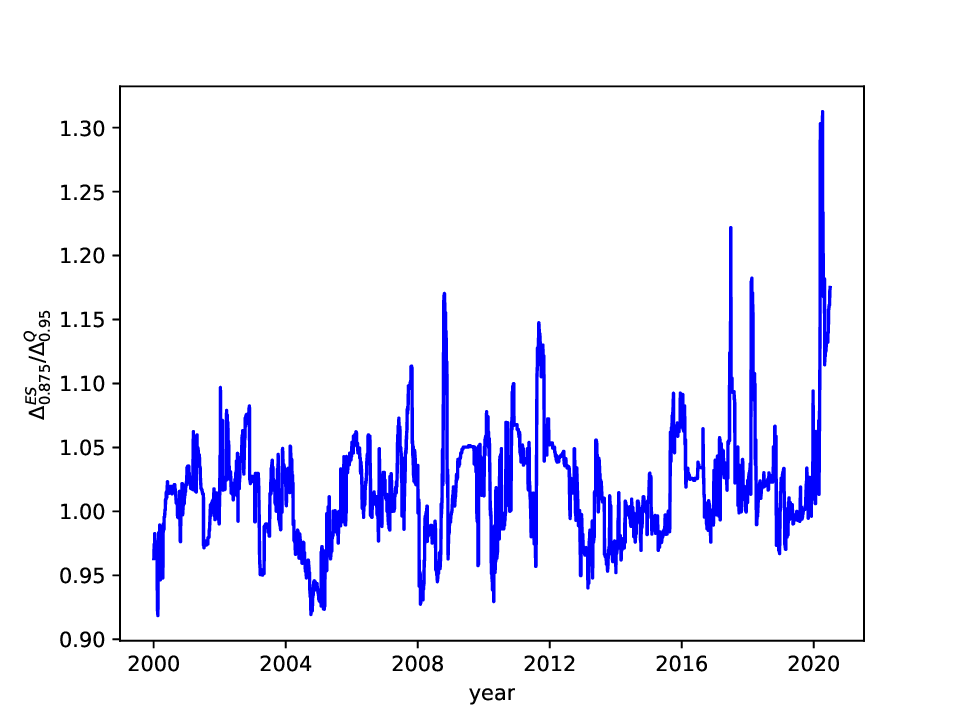}
		\label{fig1f}
	\end{subfigure}
\caption{The ratio of $\Delta^\ES_q$ to $\Delta^Q_p$ using S$\&$P 500  daily log-loss data (Jan 2000 - Jun 2020). Left: $(p,q)=(0.9,0.75)$. Right: $(p,q)=(0.95,0.875)$. }	
	\label{ratio2}
 \end{figure}

\begin{figure}[H]
	\begin{center}
     \begin{subfigure}{.4\textwidth}
		\includegraphics[scale=0.4]{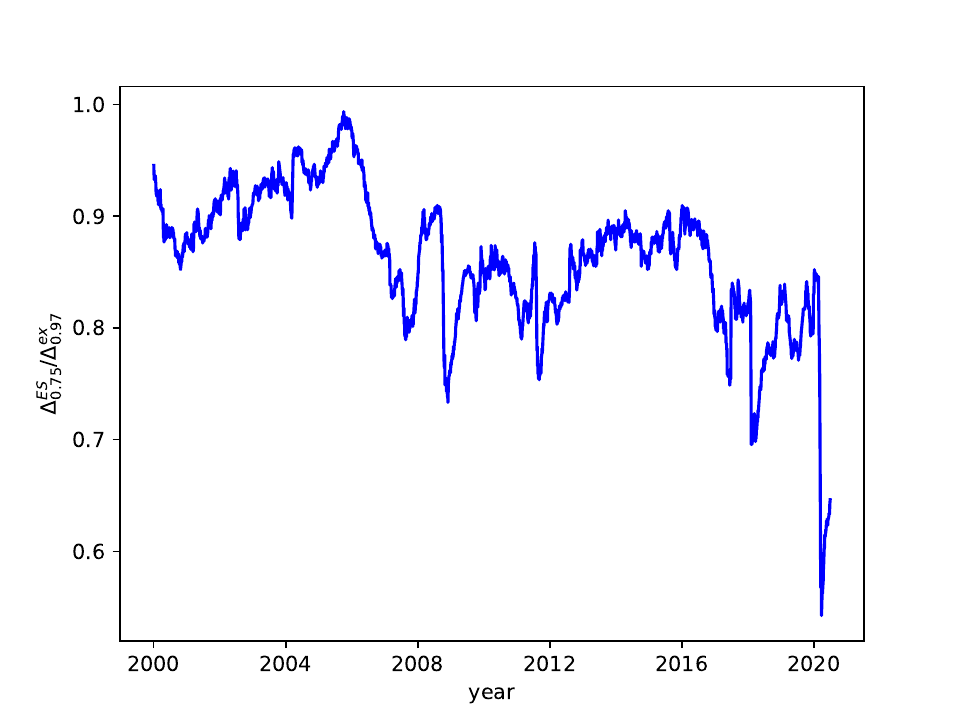}
		\label{fig1a}
	\end{subfigure}
    \begin{subfigure}{.4\textwidth}
		\includegraphics[scale=0.4]{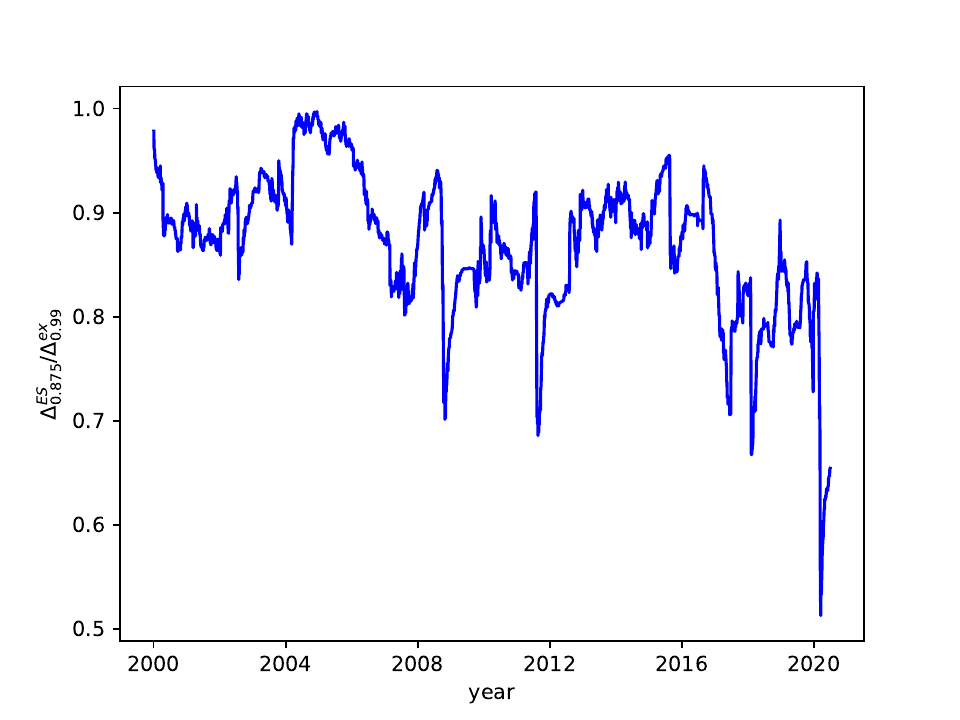}
		\label{fig1b}
	\end{subfigure}
\caption{The ratio of $\Delta^\ES_q$ to $\Delta^\ex_r$ using S$\&$P 500 daily log-loss data  (Jan 2000 - Jun 2020). Left: $(q,r)=(0.75,0.97)$. Right: $(q,r)=(0.875,0.99)$.}
	\label{ratio1}
		\end{center}
 \end{figure}
As a second empirical illustration, we compare  the distributions of the log-returns of Facebook and Berkshire Hathaway Inc.\ during the year 2020, displayed in Fig.\ \ref{fig_ord_1}. In this very peculiar year Facebook made $+33.09 \%$ with annualized volatility $46.16 \%$, and Berkshire Hathaway's made only $+2.37 \%$ with annualized volatility $35.02\%$. 
In order to check if the two distributions are comparable in one of the symmetric variability orders considered in Section \ref{sec4b}, we recall that an equivalent condition for the dilation order is given by
$$
X \leq_{\rm dil} Y \iff \ES_p(X)-\E[X] \leq \ES_p(Y)-\E[Y], \text{ for each } p \in (0,1).
$$
We see in the left panel of Fig.\ \ref{fig_ord_2} that there is an intersection point in the $\ES_p - \E$ curves, so Facebook's log-returns do not dominate Berkshire Hathaway's according to the dilation order (and vice versa). In this specific example, this is due to the presence of two  large values in the distribution of Berkshire Hathaway's daily log-returns. On the contrary, looking at the center and left panels of Fig.\ \ref{fig_ord_2}   we see that there are no intersection points, so Facebook's log-returns dominate Berkshire Hathaway's according to both the $\leq_{\rm \Delta{\text -}\ES}$ and $\leq_{\rm \Delta{\text -}\ex}$ orders. Hence, both $\leq_{\rm \Delta{\text -}\ES}$ and $\leq_{\rm \Delta{\text -}\ex}$  are able to model an ordering relation in the variability between two distributions, when the classic dilation order fails to hold, and this shows the additional flexibility of the new orders over the classic notion.

As a third example, we compare the distributions of log-returns of the S\&P500 Index in 2008 and in 2020, displayed in  Fig.\ \ref{fig_ord_3}. As in the previous example, we plot the relevant curves in Fig.\ \ref{fig_ord_4}. Here there is an intersection point both in the left and in the center panel, and no intersections in the right panel, so only the $\leq_{\rm \Delta{\text -}\ex}$ order applies. 

In order to give a first exploratory assessment of how often the various symmetric variability orders do apply, we checked the comparability of daily log-returns of the S\&P500 Index for each pair of years ranging from 2008 to 2020, for a total of $78 =13 \times 12/2$ pairs. The results are reported in Table \ref{table:ord_freq}. It turns out that in $66$ cases the $\leq_{\rm dil}$ order applies, and so as a consequence also the other two weaker orders apply. In the remaining $12$ cases, one or both of the $\leq_{\rm \Delta{\text -}\ES}$ and $\leq_{\rm \Delta{\text -}\ex}$ orders apply in $8$ cases, so when the $\leq_{\rm dil}$ order does not apply, we have a fraction of $8/12 \simeq 67 \%$  of cases in which the data can still be compared. 
Notice also that the $\leq_{\rm \Delta{\text -}\ES}$ order without the $\leq_{\rm \Delta{\text -}\ex}$ order never occurred for this dataset; however, Example \ref{ex} in Section \ref{sec4b} shows that also this situation is theoretically possible. 
\begin{table}[ht]
\centering 
\begin{tabular}{c|c|c|c|c} 
N & $\leq_{\rm dil}$ & $\leq_{\rm \Delta{\text -}\ES}$ & $\leq_{\rm \Delta{\text -}\ex}$ & pairs of years ($20$XX) \\ 
\hline 
66 & $\checkmark$ & $\checkmark$ & $\checkmark$ & all the others \\ 
6 & $\times$ & $\checkmark$ & $\checkmark$ & $(09,11), (10,15), (11,18), (12,14), (12,13), (15,18)$  \\
0 & $\times$ & $\checkmark$ & $\times$ & - \\
2 & $\times$ & $\times$ & $\checkmark$ & $(08,20), (10,18)$ \\
4 & $\times$ & $\times$ & $\times$ & $(12,16), (12,19), (13,14), (16,19)$ \\
\end{tabular}
\caption{Number of occurencies of the symmetric variability orders $\leq_{\rm dil}$, $\leq_{\rm \Delta{\text -}\ES}$ and $\leq_{\rm \Delta{\text -}\ex}$ in the $78=13 \times 12 /2$ pairs of years of daily log-returns of the S\&P500 Index, ranging from $2008$ to $2020$, and corresponding pairs. For brevity we report only years' last two digits. }
\label{table:ord_freq} 
\end{table}


\begin{figure}[htbp]
\centering 
\includegraphics[width=0.95\textwidth]{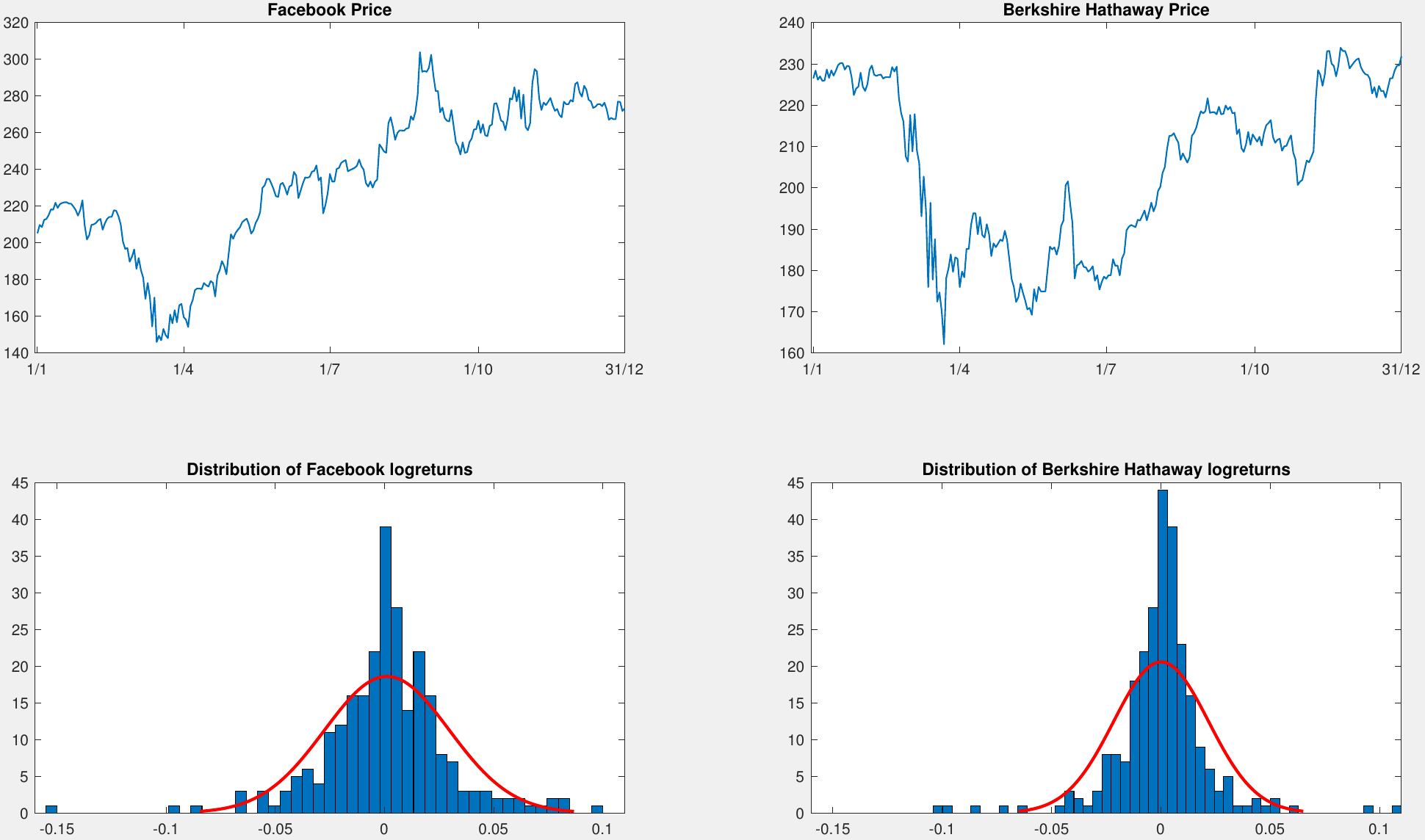}
\caption{Prices and log-return distributions of Facebook and Berkshire Hathaway in 2020.}
\label{fig_ord_1}
\end{figure}

\begin{figure}[t]
\centering 
\includegraphics[width=0.95\textwidth]{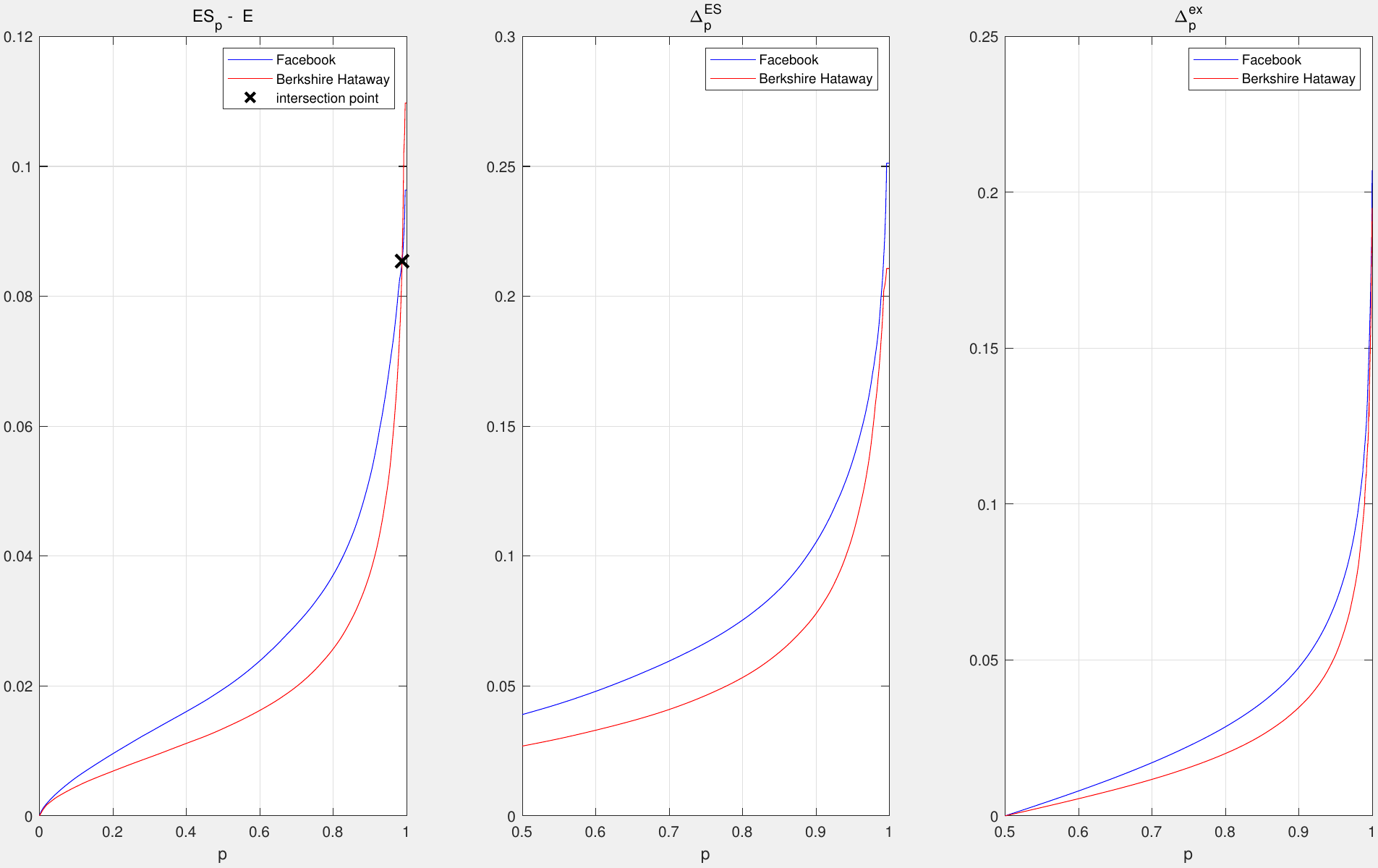}
\caption{Symmetric variability orderings between log-returns of Facebook and Berkshire Hathaway in 2020.\ To check comparability, we plot $\ES_p - \E$ (left panel), $\Delta_p^{\ES}$ (center panel) and $\Delta_p^{\ex}$ (right panel) as a function of $p$. Facebook's log-returns dominate Berkshire Hathaway's in the $\leq_{\rm \Delta{\text -}\ES}$ and $\leq_{\rm \Delta{\text -}\ex}$ orders, but not in the $\leq_{\rm dil}$ order.}
\label{fig_ord_2}
\end{figure}

\begin{figure}[htbp]
\centering 
\includegraphics[width=0.95\textwidth]{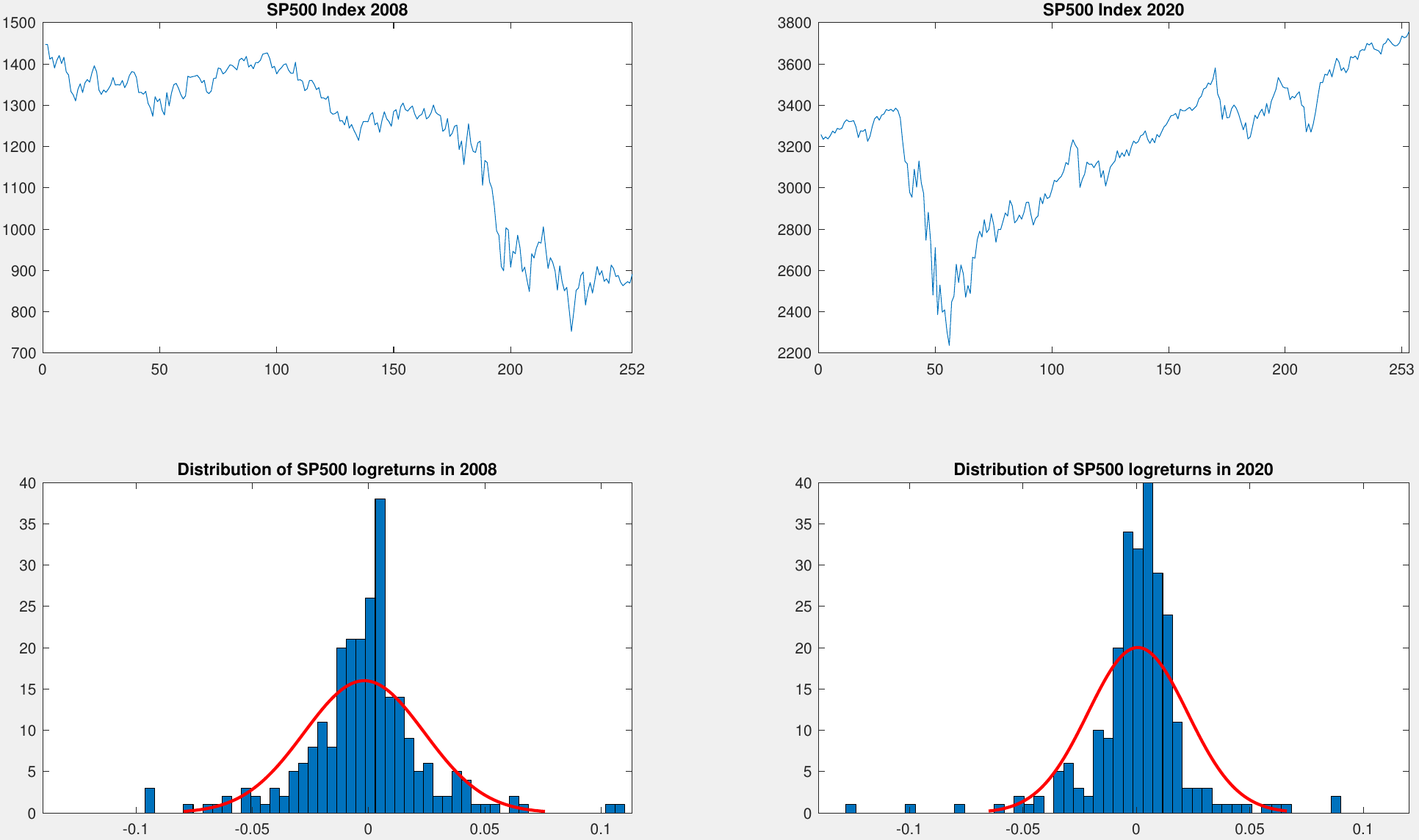}
\caption{Values and log-return distributions of the SPX Index in 2008 and 2020.}
\label{fig_ord_3}
\end{figure}

\begin{figure}[t]
\centering 
\includegraphics[width=0.95\textwidth]{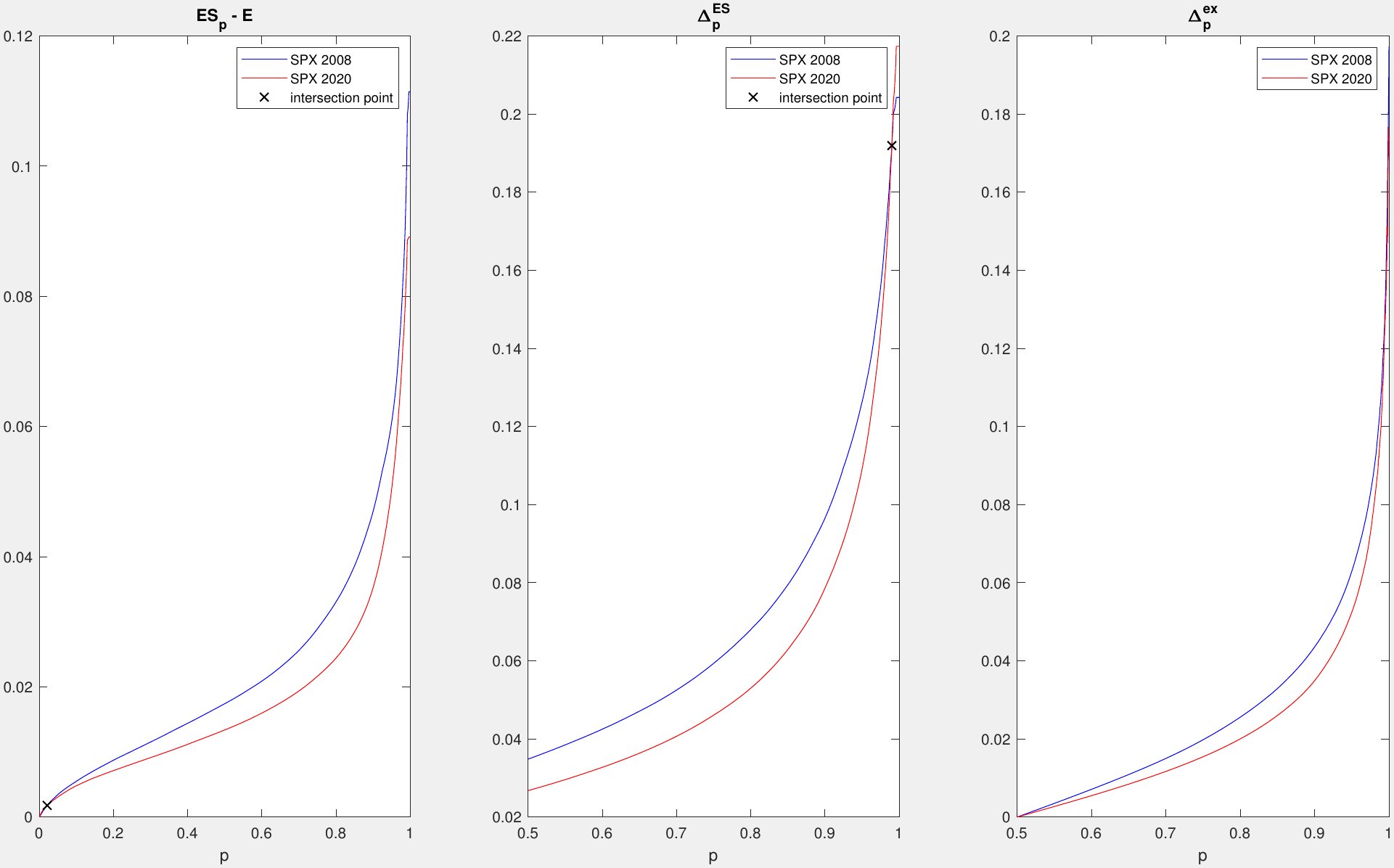}
\caption{Symmetric variability orderings between log-returns of the SPX Index in 2008 and 2020.\ To check comparability, we plot $\ES_p - \E$ (left panel), $\Delta_p^{\ES}$ (center panel) and $\Delta_p^{\ex}$ (right panel) as a function of $p$. The SPX 2008 log-returns dominate the 2020 log-returns in the $\leq_{\rm \Delta{\text -}\ex}$ order, but not in the $\leq_{\rm dil}$ and in the $\leq_{\rm \Delta{\text -}\ES}$ orders.}
\label{fig_ord_4}
\end{figure}

\section{Conclusion} \label{sec8}
In this paper, we introduce variability measures induced by three very popular parametric families of risk measures, that is, the inter-quantile, the inter-ES, and the inter-expectile 
differences.
The three classes of variability measures  enjoy many nice theoretical properties (Theorem \ref{th:0}); in particular, each of them characterizes symmetric distributions up to a location shift (Proposition \ref{th:4}).  
 We study  several desirable functional properties of general variability measures including the above three classes and many other classic ones; a grand summary is obtained in Theorem \ref{th:1} and Table \ref{tab:1}.
  The family of variability measures that satisfy a set of desirable properties  is characterized as mixtures of inter-ES differences (Theorem \ref{th:2}).
It is important to note that the three classes of variability measures introduced in this paper are well defined on $L^1$ and that each depends on a single parameter which allows  for flexible applications. This distinguishes them from other deviation measures (e.g., \cite{RUZ06}) where no parametric family is given. The empirical estimators of the inter-quantile, the inter-ES, and the inter-expectile differences can be formulated based on those of $\VaR$, $\ES$ and the expectile, and the asymptotic normality of the estimators is established (Theorem \ref{th:5}). In the financial application, we observe that the behaviour of these variability measures is similar to the corresponding 
parametric families of risk measures. However, a comparison of different ratio of the variability measures reveals that $\Delta^{\ex}$ is the most sensitive to extreme losses, and $\Delta^Q$ is the least sensitive. 

For the end-user, if  tail risk is of particular concern, then $\Delta^\ex$ may be a better variability measure to use, as it captures tail-heaviness quite effectively. However, $\Delta^\ex$ is usually cumbersome in computation and optimization because of the lack of explicit formulas in terms of quantile or distribution functions; another technical disadvantage is that $\Delta^\ex$ is not concave with respect to mixtures.  On the other hand, if robustness is more important and tail risk is not relevant, then $\Delta^Q$ is a good choice, because quantiles are easy to compute and they are generally more robust than coherent risk measures including ES and expectiles (see \cite{CDS10}). Moreover, $\Delta^Q$ is well defined on risks without a finite mean; nevertheless we should  keep in mind that $\Delta^Q$ ignores tail risk just like a quantile. Finally, $\Delta^{\ES}$ lies somewhere in between $\Delta^Q$  and $\Delta^\ex$ regarding the above considerations, which giving rise to a good compromise; further, it is the only one among the three classes that is concave with respect to mixtures (see Table \ref{tab:1}), and it is the building block for many other measures of variability (see Theorem \ref{th:2}).

In the literature, 
risk measures are commonly defined on a space of both positive and negative random variables.
For this reason, our variability measures are also defined on such spaces, and we omit  a detailed study of relative variability
measures which are defined only for positive random variables. Relative variability measures include important examples such as the relative deviation and the Gini coefficient; see Appendix \ref{app:a1}.   
By replacing    classic risk measures with relative risk measures (e.g., \cite{PQWY12}), one could define new classes of relative risk measures. 
 On the other hand, other parametric families of risk measures, such as entropic risk measures (e.g., \cite{FS16}) and RVaR (e.g., \cite{ELW18}), can also be used to design flexible variability measures.

\appendix

\section{Classic variability measures}
\label{app:a1}

Below we list some classic variability measures, which are formulated on their respective effective domains. 
\begin{enumerate}[(i)]\item  The variance ($\var$)
$$ \E[(X-\E[X])^2], ~~X \in L^2.$$
\item The standard deviation (STD):
$$
\sqrt{\var(X)}, ~~ X\in L^2.$$
\item The range ($\Delta_1$):
$$
\esssup(X)-\essinf(X), ~~ X \in L^\infty.$$
\item The mean absolute deviation (MAD):
$$
\E[|X-\E[X]|], ~~ X \in L^1.$$
\item The mean median deviation (MMD):
$$
\min_{x\in \R}\E[|X-x|]=\E[|X-Q_{1/2}(X)|], ~~ X \in L^1.$$
\item The Gini deviation (Gini-D):
$$
 \frac12\E[|X_1-X_2|] , ~~ X \in L^1,~X_1,X_2,X\mbox{ are iid}.$$
\item The relative deviation:
$$
\frac{\SD(X)}{\E[X]}, ~~ X \in L^2_+.$$
\item The Gini coefficient:
$$
\frac{\E[|X_1-X_2|]}{2\E[X]}=\frac{\Gini(X)}{\E[X]}, ~~ X \in L^1_+,~X_1,X_2,X\mbox{ are iid}.$$
\end{enumerate}
Here, $L^q_+$, $q\in [0,\infty]$ is the set of all non-negative random variables $X$ in $L^q$ with $\p(X>0)>0$.

\section{Proofs of main results}
\label{app:B}

\begin{proof}[Proof of Theorem \ref{th:0}]
\begin{enumerate}[(i)]
\item Law invariance (A1) is obvious. 
For standardization (A2), note that the risk measures $\rho\in\{Q_p,Q_p^-, \ES_p, \ES_p^-,\ex_p\}$ are all monetary (\cite{FS16}) and satisfies $\rho(m)=m$ for any constant $m$. Hence, 
for a constant $m$, $\Delta^Q_p(m)=\Delta^\ES_p(m)=\Delta^\ex_p(m)=0$. 
Positive homogeneity follows from that of $Q_p$, $Q_p^-$, $\ES_p$,  $\ES_p^-$   and $\ex_p$.
\item The effective domains of these variability measures can be easily checked from the effective domain of the corresponding risk measures.
\item Since $Q_p$ is increasing in $p$ and $Q_{1-p}^-$ is decreasing in $p$, $\Delta_p^Q$ is increasing in $p$. The same applies to $\Delta_p^{\ES}$ and $\Delta_p^{\ex}$.
\item It is well known that, for $X \in L^0$,
 $Q_p(-X) =  -Q^-_{1-p} (X)$; see e.g., \citet[(4.44)]{FS16}. Hence, 
$
\Delta^Q_p (X) = Q_p (X) + Q_{p} (-X).
$
 
 The formula for $\Delta^\ES_p$, 
$  \ES_p (X) -  \ES^-_{1-p} (X) = \ES_p (X) + \ES_p (-X),
$ follows directly from definition.

Next we show the formula for $\Delta^\ex_p$.
From \cite{NP87}, the expectile $\ex_p (X)$, for $p \in (1/2,1)$ is the unique solution $x$  to 
\begin{equation}
\label{expectile-eq}
p \E[(X-x)_+] = (1-p)\E[(X- x)_-].
\end{equation}
Hence, the expectile of $-X$ satisfies 
$$(1-p) \E[(-X-\ex_{1-p} (-X))_+] = p \E[(-X-\ex_{1-p} (-X))_-].$$
This is equivalent to 
$$p \E[(X+\ex_{1-p} (-X))_+] = (1-p) \E[(X+\ex_{1-p} (-X))_-].$$
The uniqueness of solution $x$ to \eqref{expectile-eq} implies $-\ex_{1-p} (X) = \ex_p (-X)$. Hence,
$$
\Delta^\ex_p (X)= \ex_p (X) -  \ex _{1-p} (X) = \ex_p (X) + \ex_p (-X),
$$
thus the desired formula. \qedhere
\end{enumerate}

\end{proof}

\begin{proof}[Proof of Proposition \ref{prop:eshalf}]
By definition, for $X\in L^1$,
\begin{align*}(1-p) \Delta_p^{\ES}(X) &= (1-p)   \frac{1}{1-p} \int_p^1 \left(Q_r(X)  -  Q_{1-r}(X)  \right)\d r\\ &
=  \int_{1-p}^1 \left(Q_r(X) -  Q_{1-r}(X)   \right)\d r 
- 
\int_{1-p}^p \left(Q_r(X)   -  Q_{1-r}(X)   \right)\d r
\\ &   
=   p \Delta_{1-p}^{\ES}(X) 
- 
\int_{1-p}^p Q_r(X)   \d r  + \int_{1-p}^p  Q _{r}(X)   \d r
\\& = p \Delta_{1-p}^{\ES}(X).
\end{align*}
By Theorem \ref{th:0},
\begin{align*}
\Delta_p^{\ES}(X) = \frac{1}{1-p} \int_p^1 Q_{q}(X) \d q  + \frac{1}{1-p} \int_p^1 Q_{q}(-X) \d q   
& = \frac{1}{1-p} \int_p^1 \Delta_q^Q (X) \d q.
\end{align*}
Hence, the desired statements hold.
\end{proof}

\begin{proof}[Proof of Theorem \ref{th:1}] 
We first explain some general observations on all variability measures in Table \ref{tab:1}. 
 The effective domains   and the homogeneity indices follow directly from definition. 
Continuity (B2) is implied by   $L^q$ continuity since all variability measures are finite and thus continuous on their effective domains. Symmetry (B3) and location invariance (B8) are straightforward to check, and they hold for all variability measures in Table \ref{tab:1}.

The conditions (B5)-(B7) are connected. In particular, Theorem 3 of \cite{WWW20} states that (B5)-(B7) are equivalent for \emph{distortion riskmetrics}, which are functionals satisfying (A1), (B4) and some continuity assumptions. It is well known that the inter-quantile differences and the inter-ES differences are distortion riskmetrics. 

Next, we explain that convexity (B6)  implies Cx-consistency (B5) for all variability measures we consider.  
 By Theorem 2.2 of \cite{LCLW20}, all law-invariant convex risk functionals, i.e., functionals satisfying (A1), (B6) and (B8), can be written as the supremum of a family of convex distortion riskmetrics. 
Each distortion riskmetric is Cx-consistent as stated in Theorem 3 of \cite{WWW20}, and hence (B5) is implied by (B6). 
 The only negative statement for (B5) is made for the inter-quantile difference, which is a non-convex distortion riskmetric; this is shown in Table 1 of \cite{WWW20}, which contains also a list of other examples of distortion riskmetrics with their corresponding properties.  
 Hence, the inter-quantile difference does not satisfy any of (B5)-(B7). 
 
It remains to verify (B1), (B4), (B6), (B7) for each variability measure.
\begin{enumerate}[(i)]
	\item  The following example shows that  $\Delta_p^Q$ does not satisfy (B1). 	Take $\epsilon>0$ such that $p+\epsilon<1$ and $X\sim \mathrm{Bernoulli}(1-p-\epsilon)$. Notice that $X$ is not a constant but $\Delta_p^Q(X)=Q_p(X)-Q_{1-p}^-(X)=0-0=0$.  C-additivity (B4) is satisfied since $\Delta_p^Q$ is a distortion riskmetric. (B6)-(B7) are explained above. 

	\item   $\Delta_{p}^{\ES}$, Gini-D and range  are all convex distortion riskmetrics; see Table 1 of \cite{WWW20}. Hence, they all satisfy (B4)-(B7).  Relevance (B1) can be easily verified. 
	\item   If $X$ is not a constant, by \citet[Theorem 1]{NP87}, $\ex_p$ is strictly increasing in $p\in (0,1)$, which means that $\Delta_p^{\ex}(X)=\ex_p(X)-\ex_{1-p}(X)>0$ for $p\in (1/2,1)$. By Proposition 7 of \cite{BKMR14}, $\ex_p$ is increasing in $X$, so for $|X|\le 1$, $-1\le \ex_p(X)\le 1$ for $p\in (0,1)$. Thus $\Delta_{p}^{\ex}(X)\le 2$ and  Relevance (B1)  is satisfied.  Convexity (B6) is satisfied by Theorem \ref{th:0} (iv) and  convexity of expectiles.  
	
	We show that M-concavity (B7) is not satisfied by $\Delta_p^{\ex}(X)$ via  the following example from   \cite{BBP18}. Take $p=1/10$.	Define $X$ by $\p(X=-1)=1/2$, and $\p(X=1)=1/2$; $Y$ by $\p(Y=0)=2/3$, $\p(Y=5)=1/3$. Then $\Delta_{1/10}^{\ex}(X)=-\frac{8}{5}$ and $\Delta_{1/10}^{\ex}(Y)=-\frac{800}{209}$.  
	
	Let $F=\frac{9}{10}F_X+\frac{1}{10}F_Y$ and $Z\sim F$. Then
	$$\Delta_{1/10}^{\ex}(Z)=-\frac{2531}{1311}<\frac{9}{10}\Delta_{1/10}^{\ex}(X)+\frac{1}{10}\Delta_{1/10}^{\ex}(Y)=-\frac{9524}{5225},$$
	and hence $\Delta_p^{\ex}$ is not mixture concave.
	 
 C-additivity (B4) is not satisfied since by Theorem \ref{th:2}, a variability measure satisfying (B1)-(B5) must satisfy (B7).
		
	\item For the variance, Relevance (B1) can be easily verified. 
	Variance does not satisfy (B4) since (B4) requires the homogeneity index to be $1$.
	For (B6), the variance is well known to be convex (\cite{DG85}); see also Example 2.2 of \cite{LCLW20}. 
	 	The variance satisfies M-concavity (B7) because of the well known equality
	$$
	\sigma^2(X) = \min_{x\in \R} \E[(X-x)^2],~~~X\in L^2.
	$$
	Since $\sigma^2$ is the minimum of mixture-linear functionals, we know that it is mixture concave.
%

	
	\item For STD, Relevance (B1) can be easily verified. 	C-additivity (B4) is not satisfied by STD  since STD is not additive for comonotonic random variables $X$ and $Y$ with correlation less than $1$. STD is convex (B6); see Example 2.1 of \cite{LCLW20}.  To show that STD satisfies M-concavity (B7), take $X,Y\in L^1$ and let $Z\sim \lambda F_X+(1-\lambda)F_Y$ for $\lambda\in [0,1]$. By definition,
	\begin{align*}
	&\sigma^2(Z)-(\lambda \sigma(X)+(1-\lambda)\sigma(Y))^2\\
	&=\lambda(1-\lambda)\(\E[X^2]+\E[Y^2]-2\E[X]\E[Y]-2\sigma(X)\sigma(Y)\)\\
	&=\lambda(1-\lambda)\(\E^2[X]+\sigma^2(X)+\E^2[Y]+\sigma^2(Y)-2\E[X]\E[Y]-2\sigma(X)\sigma(Y)\)\\
	&=\lambda(1-\lambda)\(\(\E[X]-\E[Y]\)^2+\(\sigma(X)-\sigma(Y)\)^2\)\ge 0,
	\end{align*}
	
	which is equivalent to $\sigma(Z)\ge \lambda\sigma(X)+(1-\lambda)\sigma(Y)$.

	\item For the mean absolute deviation (MAD), Relevance (B1) can be easily verified. 
 MAD satisfies convexity (B6), since, for   $\lambda\in [0,1]$ and $X,Y\in L^1$,
	\begin{align*}
	&\E[|\lambda X+(1-\lambda)Y-\lambda \E[X]-(1-\lambda)\E[Y]|]\\
	&\le \E[|\lambda X-\lambda\E[X]|]+\E[|(1-\lambda)(Y-\E[Y])|]=\lambda \E[|X-\E[X]|]+(1-\lambda)\E[|Y-\E[Y]|].
	\end{align*} 
	We give an example showing that MAD  does not satisfy M-concavity (B7). 
	Take $X\sim \mathrm{Bernoulli}(1/3)$, and $Y\laweq -X$.
	Let $F=\frac 12 F_X + \frac 12 F_Y$ and $Z\sim F$. 
	It is easy to calculate that $\E[X]=1/3$, $\E[Y]=-1/3$, $\E[Z]=0$,
	and 
	$\E[|X-\E[X]|=\E[|Y-\E[Y]| =4/9$.
	On the other hand,
	$$\E[|Z-\E[Z]|]= \frac{1}{2}\E[|X|] + \frac{1}2 \E[|Y|] = \frac 13.$$
	Therefore, $$\E[|Z-\E[Z]|]<\frac12  \E[|X-\E[X]| + \frac 12\E[|Y-\E[Y]|,$$
	and hence MAD is not mixture concave.
	
	C-additivity (B4) is not satisfied by MAD since by  Theorem \ref{th:2}, a variability measure satisfies (B1)-(B5) must satisfy (B7). \qedhere
\end{enumerate}

\end{proof}

\begin{proof}[Proof of Theorem \ref{th:2}]
Write the functional $\nu_\mu=  \int_0^1\Delta^{\ES}_p  \d \mu(p) $, which is the right-hand side of \eqref{eq:repES}. 
First, obviously (i) implies (ii). It is also straightforward to check that (iii) implies (i), since $\Delta^\ES_p$ for $p\in (0,1]$ satisfies (B1)-(B8) by Theorem \ref{th:1}, and so is $\nu_\mu$; the only non-trivial statement is (B2) of $\nu_\mu$ which is guaranteed by Theorem 5 of \cite{WWW20}, which shows that the representation $\nu_\mu$ belongs to a class of convex distortion riskmetrics with continuity (B2).
Below, we show (ii)$\Rightarrow$(iii).

Let $\X_\nu$ be the effective domain of $\nu$. 
Take $X\in \X_\nu$ such that $\nu(X)>0$. By (B4), $\nu(2X)= \nu(X)+\nu(X)=2\nu(X)$. Hence, the homogeneity index of $\nu$ is $1$. 

Suppose that Cx-consistency (B5) holds. Take any $X,Y\in \X_\nu$ and let $X'\laweq X$ and $Y'\laweq X$ such that $X'$ and $Y'$ are comonotonic.
It is well known that $X+Y\leq_{\rm cx} X'+Y'$; see e.g., Theorem 3.5 of \cite{R13}.
Using (B4) and (B5), we have
$$
\nu(X+Y) \le \nu (X'+Y') =\nu(X')+\nu (Y')=\nu(X)+\nu(Y).
$$
Therefore, $\nu$ is subadditive, that is,
\begin{equation}
\label{eq:subadd} \nu(X+Y) \le  \nu(X)+\nu(Y) \mbox{ for all }X,Y\in \X.
\end{equation} 
Note that convexity (B6) and homogeneity (A3) with $\alpha=1$ together also imply subadditivity.
Hence, either assuming (B5) or (B6), we get \eqref{eq:subadd}.  
It follows from  \eqref{eq:subadd} and (B1) that  there exists $\beta>0$  such that
$\nu(Y)-\nu(X) \le \nu (Y-X) \le \beta \Vert Y-X\Vert_\infty   $ where $\Vert Y-X \Vert_\infty$ is the essential supremum of $|Y-X|$.
Hence,  $\nu$ is uniformly continuous with respect to the supremum norm. 
Moreover,  as a consequence of (B1), (A3) and \eqref{eq:subadd}, $\X_\nu$ is a convex cone that contains $L^\infty$.

 Theorem  1  of \cite{WWW20} states that a real functional on a convex cone that is 
 uniformly continuous with respect to the supremum norm, law-invariant, and satisfying (B2) and (B4) 
is  a   distortion riskmetric in the sense of that paper; see \eqref{eq:distortionrep} below.
  Further, Theorem 3 of \cite{WWW20} says that each of (B5)-(B7) is equivalent to the convexity of a distortion riskmetric. 
Hence, $\nu$ is a convex distortion riskmetric on $\X_\nu\cap L^1$. 
 Theorem 5 of \cite{WWW20} gives a representation of 
 convex distortion riskmetrics; that is, 
$\nu$ has a representation, for some finite measures $\mu_1$ and $\mu_2$,
\begin{equation}
\label{eq:distortionrep}
 \nu (X) = \int_0^1 \ES_p (X) \d \mu_1(p) + \int_0^1 \ES_{p}(-X) \d \mu_2(p), ~~~X\in \X_\nu\cap L^1.
\end{equation}
By symmetry (B3), we know 
 $$
 \nu (X)= \nu (-X)  = \int_0^1 \ES_p (X) \d \mu_2(p) + \int_0^1 \ES_{p}(-X) \d \mu_1(p), ~~~X\in \X_\nu\cap L^1.
 $$
 Hence, we can take $\mu=(\mu_1+\mu_2)/2$,
 and get
  \begin{align*}
\nu(X) = \int_0^1\Delta^{\ES}_p (X)\d \mu(p),~~X\in \X_\nu\cap L^1.
\end{align*}
 Relevance (B1)   implies $\mu\ne 0$, which in turn implies $\X_\nu \subset L^1$, as the effective domain of $\Delta_p^\ES$ is $L^1$ for $p\in (0,1)$. 
 Hence, the two functionals $ \nu$  and $ \nu_\mu $ coincide on $\X_\nu$ which contains $L^\infty$.
 Also note that both 
 $ \nu$  and  $ \nu_\mu $   satisfy  continuity (B2),
 and hence one can approximate any random variable outside $\X_\nu$ with truncated random variables, and obtain that $ \nu$  and  $ \nu_\mu $  also coincide on $\X$. 
\end{proof}

\begin{proof}[Proof of Proposition \ref{th:4}]
\begin{enumerate}[(i)]
	\item If $X$ has a symmetric distribution, then by Theorem \ref{th:0} (iv),  we have $$\Delta_p^Q(X)=Q_p(X)-Q_{1-p}^-(X)=-Q_{1-p}^-(-X)-Q_{1-p}^-(X)=-2Q_{1-p}^-(X).$$  	Assume $X_1$ and $X_2$ are symmetric distributions with finite $\Delta_p^Q(X_1)=\Delta_p^Q(X_2)$ for $p\in (\frac{1}{2},1)$. It follows that $Q_p^-(X_1)=Q_p^-(X_2)$ for $p\in (0, \frac{1}{2})$. By the left-continuity of the left-quantile, $Q_{1/2}^-(X_1)=Q_{1/2}^-(X_2)$.
By symmetry of the distribution of $X$, we have $Q_{p}^-(X_1)=Q_{p}^-(X_2)$ almost every $p$, and thus $X_1$ and $X_2$ have the same distribution. 
\item If $X$ has a symmetric distribution, then similarly to (i),  we have $\Delta_{p}^{\ES}(X)=-2\ES_{1-p}^-(X).$
 	Assume that $X_1$ and $X_2$ have symmetric distributions with finite $\Delta_p^\ES(X_1)=\Delta_p^\ES(X_2)$ for $p\in (\frac{1}{2},1)$. It follows that  for $p \in (0, \frac{1}{2})$, $\ES_p^-(X_1)=\ES_p^-(X_2)$ holds, which means 
\begin{equation}\label{eq:es-equal}
 \int_0^p Q_r(X_1)\d r= \int_0^p Q_r(X_2)\d r.
\end{equation} 
 	By taking a derivative of both sides of \eqref{eq:es-equal} with respect to $p$, we get
 	$$Q_p(X_1)=Q_p(X_2)$$
	at all common continuity points $p$ of $p\mapsto Q_p(X_1)$ and $p\mapsto Q_p(X_1)$.  Since both functions are right-continuous, we know that the two functions are identical. 
 	This argument can be applied to any $p\in (0,\frac{1}{2})$. Similarly to part (i), we conclude that $X_1$ and $X_2$ have the same distribution.
\item  
 If $X$ has a symmetric distribution, then similarly to (i),  we have $$\Delta_{p}^{\ex}(X)=2\ex_p(X)=-2\ex_{1-p}(X).$$
 	Suppose $X_1$ and $X_2$ have symmetric distributions with finite $\Delta_{p}^{\ex}(X_1)=\Delta_{p}^{\ex}(X_2)$ for $p\in (\frac{1}{2},1)$. Then $\ex_p(X_1)=\ex_p(X_2)$ for $p\in (0,\frac{1}{2})\cup (\frac{1}{2},1).$ By symmetry, we observe that $\E[X_1]=\E[X_2]=0$, which means $\ex_{\frac{1}{2}}(X_1)=\ex_{\frac{1}{2}}(X_2)=0$, so $\ex_p(X_1)=\ex_p(X_2)$ for $p\in (0,1)$.
 
 The expectile has alternative definitions from \cite{NP87},
 $$\ex_p(X)=\E[X]+\frac{2p -1}{1-p}\E[(X-\ex_p(X))_+],$$
which leads to
 	$$\E[(X_1-\ex_p(X_1))_+]=\E[(X_2-\ex_p(X_2))_+].$$
	Since $\ex_p(X)$ is continuous in $p$ and takes all values in the range of $X$, we know 
	$$\E[(X_1-x)_+]=\E[(X_2-x)_+]$$
	for all $x\in \R$, implying that the distributions of $X_1$ and $X_2$ are identical. \qedhere
\end{enumerate} 
\end{proof}

\begin{proof}[Proof of Proposition \ref{th:orders}]
(i), (ii), (iii) follow immediately, respectively from location invariance, positive homogeneity of order  $1$ and symmetry of $\Delta_p^{\ES}$ and $\Delta_p^{\ex}$, while (iv) follows immediately from the second part of the thesis of Proposition \ref{prop:eshalf}.
\begin{itemize}
\item [(v)] By passing if necessary to the random variables $\tilde{X}=X-\E[X]$ and $\tilde{Y}=Y-\E[Y]$, from (i) we can assume without loss of generality that $\E[X]=\E[Y]=0$. Then $X\leq_{\rm dil} Y \Rightarrow X \leq_{\rm cx} Y$, and the thesis follows from Cx-consistency of $\Delta_p^{\ES}$ and $\Delta_p^{\ex}$, for each $p \in (1/2,1)$. 
\item [(vi)] As in (v), we can assume  without loss of generality  that $\E[X]=\E[Y]=0$. Then $\Delta_p^{\ES}(X)=2\ES_p(X)$, so $X\leq_{\rm \Delta{\text -}\ES}Y \Rightarrow \ES_p(X) \leq \ES_p(Y) \Rightarrow \int_p^1 Q_r(X) \d r \leq \int_p^1 Q_r(Y) \d r$, for each $p \in (1/2,1)$. From symmetry and the assumption $\E[X]=\E[Y]=0$ it follows that the same inequality holds also for each $p \in (0,1/2)$, that implies $X\leq_{\rm cx}Y$ by Theorem 3.A.5 in \cite{SS07}. Similarly, under symmetry $\Delta_p^{\ex}(X)=2\ex_p(X)$, so $X\leq_{\rm \Delta{\text -}\ex}Y \Rightarrow \ex_p(X) \leq \ex_p(Y)$ for each $p \in (1/2,1)$, and since $\ex_p(X)=\ex_p(-X)=-\ex_{1-p}(X)$, the opposite inequality holds for $p \in (0,1/2)$. By reasoning as in the proof of Theorem 12 of \cite{BKM18}, it follows that $\pi_X(x) \leq \pi_Y(x)$ for each $x \in \R$, where $\pi_X (x):=\E[(X-x)_+]$ and $\pi_Y (x):=\E[(Y-x)_+]$ are the usual stop-loss transforms of $X$ and $Y$; the thesis then follows from Theorem 3.A.1 of \cite{SS07}.  \qedhere
\end{itemize}
\end{proof}

\begin{proof}[Proof of Theorem \ref{th:5}]
\begin{enumerate}[(i)] 
    \item     Let $\widehat{Q}_p(n)$, $\widehat{\ES}_p(n)$, and $\widehat{\ex}_p(n)$  be the empirical estimators of $Q_p(X)$, $\ES_p(X)$, and $\ex_p(X)$ based on $n$ sample data points. It is well known (e.g., \cite{B66})  that $\widehat{Q}_r(n)\pto Q_r(X)$ at each $r$ of continuous point of $Q_r(X)$, which implies   $\widehat{\Delta}_p^Q(n)\pto \Delta_p^Q(X)$ under assumption (R).  Since $\ES_p$ and $\ex_p$ are law-invariant convex risk measures, by Theorem 2.6 of \cite{KSZ14},  $\widehat{\ES}_r(n)\pto \ES_r(X)$ and $\widehat{\ex}_r(n)\pto \ex_r(X)$ for each $r$. Hence we have $\widehat{\Delta}_p^{\ES}(n)\pto \Delta_p^{\ES}(X)$ and $\widehat{\Delta}_p^{\ex}(n)\pto \Delta_p^{\ex}(X)$.
      
   \item   By Proposition 1  of \citet[p.640]{SW09}, if assumption (R) is satisfied, then we have
   \begin{align}\label{conv_Q}
      \sqrt{n}\(\widehat{Q}_{p}(n) -  Q_{p}(X)\)\dto \frac{B_p}{g(p)}.
   \end{align}  where $B_p$ is a standard Brownian bridge.   
    With assumption (R),   $Q_p(X)=Q_p^-(X)$. Hence,
 \begin{align*}
   \sqrt{n}\(\widehat{\Delta}_p^Q(n) -  \Delta_p^Q(X)\)\dto  \frac{B_p}{g(p)} -  \frac{B_{1-p}}{g(1-p)} ,
   \end{align*} 
   which has a Gaussian distribution.
   Using the covariance property of the Brownian bridge, that is, $\cov[B_t, B_s] = s-st$ for $s<t$, we have   
   $$\cov\[\frac{B_p}{g(p)}, \frac{B_{1-p}}{g(1-p)}\] = \frac{(1-p)^2}{g(p)g(1-p)}.$$
    Therefore,  $
   \sqrt{n} (\widehat{\Delta}_p^Q(n) -  \Delta_p^Q(X) )\dto \mathrm{N}(0,\sigma_Q^2),
   $
   where $\sigma_Q^2$ is  in \eqref{eq:asymQ}, namely, $$\sigma_Q^2= \frac{ p(1-p)}{g^2(p)}  + \frac{ p(1-p)}{g^2(1-p)}   -2 \frac{(1-p)^2}{g(p)g(1-p)}.$$
   
 Next, we address the inter-ES difference. Applying the convergence in \eqref{conv_Q} to $\ES_p$, we obtain  
   $$\sqrt{n}\(\widehat{\ES}_p(n)-\ES_p(X)\)\dto \frac{1}{1-p}\int_p^1 \frac{B_s}{g(s)}\d s,$$
   and    thus
   $$ \sqrt{n} (\widehat \Delta_p^\ES (n) - \Delta_p^\ES (X)) \dto \frac{1}{1-p}\int_p^1 \frac{B_s}{g(s)}\d s-\frac{1}{1-p}\int_0^{1-p} \frac{B_s}{g(s)}\d s.$$ 
Note that  
\begin{align*}
 \var \[\frac{1}{1-p}\int_p^1 \frac{B_s}{g(s)}\d s\]&=\E\[\frac{1}{(1-p)^2}\int_p^1\int_p^1\frac{B_sB_t}{g(s)g(t)}\d t\d s\]\\
&=\frac{1}{(1-p)^2}\int_p^1\int_p^1 \frac{s\wedge t-st}{g(s)g(t)}\d t \d s,
\end{align*}
and 
    \begin{align*}
 \frac{1}{(1-p)^2}\cov \[ \int_{p}^{1} \frac{1}{g(t)} B_t \d t, \int_{0}^{1-p} \frac{1}{g(t)} B_t \d t \]
    & = \frac{1}{(1-p)^2} \int_{p}^{1} \int_{0}^{1-p} \frac{s\wedge t-st}{g(t)g(s)}  \d t  \d s.
    \end{align*} 
           Hence, 
    $ \sqrt{n} (\widehat \Delta_p^\ES (n) - \Delta_p^\ES (X)) \dto \mathrm{N}(0,\sigma^2_\ES),
    $
    with $  \sigma_\ES^2$ given in \eqref{eq:asymES}, namely,
$$    \sigma_\ES^2  =   \frac{1}{(1-p)^2} \left(\int_{[p,1]^2\cup [0,1-p]^2}  -2 \int_{[p,1]\times [0,1-p]}  \right)\frac{s\wedge t-st}{g(t)g(s)}\d t \d s  .$$

For the inter-expectile difference, we use Theorem 3.2 of \cite{KZ17}. The conditions for this theorem are satisfied in our setting noting that $X\in L^{2+\delta}$; see Remark 3.4 of \cite{KZ17}. We obtain, for $p\in (1/2,1)$,
    $$\sqrt{n} (\widehat \ex_{p}(n) - \ex_{p}(X) ) \rightarrow \mathrm{N}(0, s_p^\ex)$$
  where  for $r\in\{1-p,p\}$,
    $$s_r^\ex = \int_{-\infty}^\infty\int_{-\infty}^\infty f^\ex_{r,F}(t)f^\ex_{r,F}(s)F(t \wedge s)(1-F(t \vee s))\d t\d s,$$
    and 
    $$ f^\ex_{r,F}(t)=  \frac{ (1-r) \mathbb \id_{\{t\le \ex_r(X)\}}  +
r \mathbb \id_{\{t>\ex_r(X)\}}} 
{(1-2r)F(\ex_r(X))+ r}.$$
    Similar arguments as above lead to  
    $$ \sqrt{n} (\widehat \Delta_p^\ex (n) - \Delta_p^\ex (X)) \dto \mathrm{N}(0,s_p^\ex + s_{1-p}^\ex - 2c_p^{\ex}), $$
where 
$$c^{\ex}_p = \int_{-\infty}^\infty\int_{-\infty}^\infty f^{\ex}_{p,F}(t)  f^{\ex}_{1-p,F}(s)F(t \wedge s)(1-F(t \vee s))\d t \d s.$$
\end{enumerate} 
This completes the proof. 
 \end{proof}
 
 \subsubsection*{Acknowledgements}
 
 We thank an Editor, an Associate Editor, and two anonymous referees for helpful comments. 
    T.~Fadina and R.~Wang are supported by the Natural Sciences and Engineering Research Council of Canada (RGPIN-2018-03823, RGPAS-2018-522590).


\end{document}